\documentclass{article}      
\usepackage{amsmath,amssymb,graphicx,amsfonts,placeins,cite,balance,cite,bm,comment}

\usepackage{amsthm}
\usepackage{marginnote}
\usepackage{authblk}
\newtheorem{rem}{Remark}

\usepackage{tikz}
\usepackage{xifthen}
\usetikzlibrary{intersections}
\usepackage{tkz-euclide}

\pgfdeclarelayer{bg}    
\pgfsetlayers{bg,main}  
\definecolor{beige}{RGB}{245, 245, 220}
\definecolor{darkgrey}{RGB}{65, 65, 65}
\definecolor{lightgrey}{RGB}{250, 250, 250}

\newcommand*\smallcircled[1]{\tikz[baseline=(char.base)]{
            \node[shape=circle,draw,inner sep=0pt,minimum width=3.0mm] (char) {\footnotesize{#1}};}}

\usetikzlibrary{calc, arrows, fit, positioning, patterns, decorations.pathreplacing, shapes}
\tikzstyle{dash} = [dashed, -latex,>=latex]
\tikzstyle{line} = [draw, -latex,>=latex]
\tikzstyle{box} = [draw, minimum size=.8cm]
\tikzstyle{detbox} = [draw, minimum size=.5cm]
\tikzstyle{roundbox} = [draw, circle, inner sep=0pt, minimum size=3mm]
\tikzstyle{clamped} = [draw, fill=black, minimum size=0.15cm]
\tikzstyle{msgcircle} = [shape=circle, draw, inner sep=0pt, minimum size=4mm, fill=white]
\tikzstyle{darkmsgcircle} = [shape=circle, draw, inner sep=0pt, minimum size=4mm, fill=darkgrey, text=white, font=\bfseries]
\tikzstyle{msgdoublecircle} = [shape=circle, double, double distance=1.5pt, draw, inner sep=0pt, minimum size=5mm, fill=white]
\tikzstyle{darkmsgdoublecircle} = [shape=circle, double, double distance=1.5pt, draw, inner sep=0pt, minimum size=5mm, fill=darkgrey, text=white, font=\bfseries]

\newcommand{\msg}[6]{
      \ifthenelse{\isin{#1}{left} \AND \isin{#2}{down}}{
            \coordinate (anchor) at ($({#3})!{#5}!({#4})$);
            \node[msgcircle, xshift=-5.5mm] at (anchor) {#6};
            \node[xshift=-1.5mm] at (anchor) {$\downarrow$};
      }{}
      \ifthenelse{\isin{#1}{right} \AND \isin{#2}{down}}{
            \coordinate (anchor) at ($({#3})!{#5}!({#4})$);
            \node[msgcircle, xshift=5.5mm] at (anchor) {#6};
            \node[xshift=1.5mm] at (anchor) {$\downarrow$};
      }{}

      \ifthenelse{\isin{#1}{down} \AND \isin{#2}{right}}{
            \coordinate (anchor) at ($({#3})!{#5}!({#4})$);
            \node[msgcircle, yshift=-6.0mm] at (anchor) {#6};
            \node[yshift=-2.0mm] at (anchor) {$\rightarrow$};
      }{}
      \ifthenelse{\isin{#1}{up} \AND \isin{#2}{right}}{
            \coordinate (anchor) at ($({#3})!{#5}!({#4})$);
            \node[msgcircle, yshift=6.0mm] at (anchor) {#6};
            \node[yshift=2.0mm] at (anchor) {$\rightarrow$};
      }{}

      \ifthenelse{\isin{#1}{down} \AND \isin{#2}{left}}{
            \coordinate (anchor) at ($({#3})!{#5}!({#4})$);
            \node[msgcircle, yshift=-6.0mm] at (anchor) {#6};
            \node[yshift=-2.0mm] at (anchor) {$\leftarrow$};
      }{}
      \ifthenelse{\isin{#1}{up} \AND \isin{#2}{left}}{
            \coordinate (anchor) at ($({#3})!{#5}!({#4})$);
            \node[msgcircle, yshift=6.0mm] at (anchor) {#6};
            \node[yshift=2.0mm] at (anchor) {$\leftarrow$};
      }{}

      \ifthenelse{\isin{#1}{left} \AND \isin{#2}{up}}{
            \coordinate (anchor) at ($({#3})!{#5}!({#4})$);
            \node[msgcircle, xshift=-5.5mm] at (anchor) {#6};
            \node[xshift=-1.5mm] at (anchor) {$\uparrow$};
      }{}
      \ifthenelse{\isin{#1}{right} \AND \isin{#2}{up}}{
            \coordinate (anchor) at ($({#3})!{#5}!({#4})$);
            \node[msgcircle, xshift=5.5mm] at (anchor) {#6};
            \node[xshift=1.5mm] at (anchor) {$\uparrow$};
      }{}
}

\newcommand{\darkmsg}[6]{
      \ifthenelse{\isin{#1}{left} \AND \isin{#2}{down}}{
            \coordinate (anchor) at ($({#3})!{#5}!({#4})$);
            \node[darkmsgcircle, xshift=-5.5mm] at (anchor) {#6};
            \node[xshift=-1.5mm] at (anchor) {$\downarrow$};
      }{}
      \ifthenelse{\isin{#1}{right} \AND \isin{#2}{down}}{
            \coordinate (anchor) at ($({#3})!{#5}!({#4})$);
            \node[darkmsgcircle, xshift=5.5mm] at (anchor) {#6};
            \node[xshift=1.5mm] at (anchor) {$\downarrow$};
      }{}

      \ifthenelse{\isin{#1}{down} \AND \isin{#2}{right}}{
            \coordinate (anchor) at ($({#3})!{#5}!({#4})$);
            \node[darkmsgcircle, yshift=-6.0mm] at (anchor) {#6};
            \node[yshift=-2.0mm] at (anchor) {$\rightarrow$};
      }{}
      \ifthenelse{\isin{#1}{up} \AND \isin{#2}{right}}{
            \coordinate (anchor) at ($({#3})!{#5}!({#4})$);
            \node[darkmsgcircle, yshift=6.0mm] at (anchor) {#6};
            \node[yshift=2.0mm] at (anchor) {$\rightarrow$};
      }{}

      \ifthenelse{\isin{#1}{down} \AND \isin{#2}{left}}{
            \coordinate (anchor) at ($({#3})!{#5}!({#4})$);
            \node[darkmsgcircle, yshift=-6.0mm] at (anchor) {#6};
            \node[yshift=-2.0mm] at (anchor) {$\leftarrow$};
      }{}
      \ifthenelse{\isin{#1}{up} \AND \isin{#2}{left}}{
            \coordinate (anchor) at ($({#3})!{#5}!({#4})$);
            \node[darkmsgcircle, yshift=6.0mm] at (anchor) {#6};
            \node[yshift=2.0mm] at (anchor) {$\leftarrow$};
      }{}

      \ifthenelse{\isin{#1}{left} \AND \isin{#2}{up}}{
            \coordinate (anchor) at ($({#3})!{#5}!({#4})$);
            \node[darkmsgcircle, xshift=-5.5mm] at (anchor) {#6};
            \node[xshift=-1.5mm] at (anchor) {$\uparrow$};
      }{}
      \ifthenelse{\isin{#1}{right} \AND \isin{#2}{up}}{
            \coordinate (anchor) at ($({#3})!{#5}!({#4})$);
            \node[darkmsgcircle, xshift=5.5mm] at (anchor) {#6};
            \node[xshift=1.5mm] at (anchor) {$\uparrow$};
      }{}
}

\newcommand{\bwmsg}[6]{
      \ifthenelse{\isin{#1}{left} \AND \isin{#2}{down}}{
            \coordinate (anchor) at ($({#3})!{#5}!({#4})$);
            \node[msgdoublecircle, xshift=-5.5mm] at (anchor) {#6};
            \node[xshift=-1.5mm] at (anchor) {$\downarrow$};
      }{}
      \ifthenelse{\isin{#1}{right} \AND \isin{#2}{down}}{
            \coordinate (anchor) at ($({#3})!{#5}!({#4})$);
            \node[msgdoublecircle, xshift=5.5mm] at (anchor) {#6};
            \node[xshift=1.5mm] at (anchor) {$\downarrow$};
      }{}

      \ifthenelse{\isin{#1}{down} \AND \isin{#2}{right}}{
            \coordinate (anchor) at ($({#3})!{#5}!({#4})$);
            \node[msgdoublecircle, yshift=-6.0mm] at (anchor) {#6};
            \node[yshift=-2.0mm] at (anchor) {$\rightarrow$};
      }{}
      \ifthenelse{\isin{#1}{up} \AND \isin{#2}{right}}{
            \coordinate (anchor) at ($({#3})!{#5}!({#4})$);
            \node[msgdoublecircle, yshift=6.0mm] at (anchor) {#6};
            \node[yshift=2.0mm] at (anchor) {$\rightarrow$};
      }{}

      \ifthenelse{\isin{#1}{down} \AND \isin{#2}{left}}{
            \coordinate (anchor) at ($({#3})!{#5}!({#4})$);
            \node[msgdoublecircle, yshift=-6.0mm] at (anchor) {#6};
            \node[yshift=-2.0mm] at (anchor) {$\leftarrow$};
      }{}
      \ifthenelse{\isin{#1}{up} \AND \isin{#2}{left}}{
            \coordinate (anchor) at ($({#3})!{#5}!({#4})$);
            \node[msgdoublecircle, yshift=6.0mm] at (anchor) {#6};
            \node[yshift=2.0mm] at (anchor) {$\leftarrow$};
      }{}

      \ifthenelse{\isin{#1}{left} \AND \isin{#2}{up}}{
            \coordinate (anchor) at ($({#3})!{#5}!({#4})$);
            \node[msgdoublecircle, xshift=-5.5mm] at (anchor) {#6};
            \node[xshift=-1.5mm] at (anchor) {$\uparrow$};
      }{}
      \ifthenelse{\isin{#1}{right} \AND \isin{#2}{up}}{
            \coordinate (anchor) at ($({#3})!{#5}!({#4})$);
            \node[msgdoublecircle, xshift=5.5mm] at (anchor) {#6};
            \node[xshift=1.5mm] at (anchor) {$\uparrow$};
      }{}
}

\newcommand{\bwdarkmsg}[6]{
      \ifthenelse{\isin{#1}{left} \AND \isin{#2}{down}}{
            \coordinate (anchor) at ($({#3})!{#5}!({#4})$);
            \node[darkmsgdoublecircle, xshift=-5.5mm] at (anchor) {#6};
            \node[xshift=-1.5mm] at (anchor) {$\downarrow$};
      }{}
      \ifthenelse{\isin{#1}{right} \AND \isin{#2}{down}}{
            \coordinate (anchor) at ($({#3})!{#5}!({#4})$);
            \node[darkmsgdoublecircle, xshift=5.5mm] at (anchor) {#6};
            \node[xshift=1.5mm] at (anchor) {$\downarrow$};
      }{}

      \ifthenelse{\isin{#1}{down} \AND \isin{#2}{right}}{
            \coordinate (anchor) at ($({#3})!{#5}!({#4})$);
            \node[darkmsgdoublecircle, yshift=-6.0mm] at (anchor) {#6};
            \node[yshift=-2.0mm] at (anchor) {$\rightarrow$};
      }{}
      \ifthenelse{\isin{#1}{up} \AND \isin{#2}{right}}{
            \coordinate (anchor) at ($({#3})!{#5}!({#4})$);
            \node[darkmsgdoublecircle, yshift=6.0mm] at (anchor) {#6};
            \node[yshift=2.0mm] at (anchor) {$\rightarrow$};
      }{}

      \ifthenelse{\isin{#1}{down} \AND \isin{#2}{left}}{
            \coordinate (anchor) at ($({#3})!{#5}!({#4})$);
            \node[darkmsgdoublecircle, yshift=-6.0mm] at (anchor) {#6};
            \node[yshift=-2.0mm] at (anchor) {$\leftarrow$};
      }{}
      \ifthenelse{\isin{#1}{up} \AND \isin{#2}{left}}{
            \coordinate (anchor) at ($({#3})!{#5}!({#4})$);
            \node[darkmsgdoublecircle, yshift=6.0mm] at (anchor) {#6};
            \node[yshift=2.0mm] at (anchor) {$\leftarrow$};
      }{}

      \ifthenelse{\isin{#1}{left} \AND \isin{#2}{up}}{
            \coordinate (anchor) at ($({#3})!{#5}!({#4})$);
            \node[darkmsgdoublecircle, xshift=-5.5mm] at (anchor) {#6};
            \node[xshift=-1.5mm] at (anchor) {$\uparrow$};
      }{}
      \ifthenelse{\isin{#1}{right} \AND \isin{#2}{up}}{
            \coordinate (anchor) at ($({#3})!{#5}!({#4})$);
            \node[darkmsgdoublecircle, xshift=5.5mm] at (anchor) {#6};
            \node[xshift=1.5mm] at (anchor) {$\uparrow$};
      }{}
}

\newcommand{\T}{\operatorname{T}}
\newcommand{\NW}[1]{\mathcal{N}_W\!\left({#1}\right)}
\newcommand{\NV}[1]{\mathcal{N}_V\!\left({#1}\right)}
\renewcommand{\d}[1]{\operatorname{d}\!{#1}}
\newtheorem{thm}{\protect\theoremname}
\providecommand{\theoremname}{Theorem}
\newtheorem{cor}{\protect\corollaryname}
\providecommand{\corollaryname}{Corollary}

\newcommand{\E}{\mathbb{E}}

\newcommand{\KL}{\operatorname{D}}

\newcommand{\yf}{\overline{y}}
\newcommand{\yp}{\underline{y}}
\newcommand{\uf}{\overline{u}}
\newcommand{\up}{\underline{u}}
\newcommand{\xf}{\overline{x}}
\newcommand{\xp}{\underline{x}}

\newcommand{\pif}{\overline{\pi}}

\newcommand{\pf}{\overline{p}}
\newcommand{\pp}{\underline{p}}

\newcommand{\mode}{\operatorname{mode}}

\allowdisplaybreaks
 
\title{Application of the Free Energy Principle to Estimation and Control}

\author[$$]{Thijs van de Laar\thanks{TvdL acknowledges the support from GN Hearing A/S and the Netherlands Organization for Scientific Research, project number 13925.}}
\author[$$]{Ay\c ca \"Oz\c celikkale\thanks{A\"O acknowledges the support from the Swedish Research Council, Grant 2015-04011.}}
\author[$$]{Henk Wymeersch\thanks{HW acknowledges the support from the Swedish Research Council, Grant 2018-03701. The authors thank Themistoklis Charalambous for valuable comments and Magnus Koudahl for the stimulating discussions.}}

\affil[$*$]{Dept. of Electrical Engineering, Eindhoven University of Technology, The Netherlands}
\affil[$\dagger$]{Dept. of Electrical Engineering, Uppsala University, Sweden}
\affil[$\ddagger$]{Dept. of Electrical Engineering, Chalmers University of Technology, Sweden}

\begin{document}

\maketitle

\begin{abstract}
Based on a generative model (GM) and beliefs over hidden states, the free energy principle (FEP) enables an agent to sense and act by minimizing a free energy bound on Bayesian surprise. Inclusion of prior beliefs in the GM about desired states leads to active inference (ActInf). In this work, we aim to reveal connections between ActInf and stochastic optimal control. We reveal that, in contrast to standard cost and constraint-based solutions, ActInf gives rise to a minimization problem that includes both an information-theoretic surprise term and a model-predictive control cost term. We further show under which conditions both methodologies yield the same solution for estimation and control. For a case with linear Gaussian dynamics and a quadratic cost, we illustrate the performance of ActInf under varying system parameters and compare to classical solutions for estimation and control.
\end{abstract}
\textbf{Index terms} --- \emph{Active Inference, Stochastic Optimal Control, Message Passing, Factor Graphs}

\clearpage
\section{Introduction}
\label{sec:intro}

Bayesian graphical models (BGMs) constitute an important family of tools in signal processing, as they allow learning of models as well as inference of hidden states in a unified way, often with low complexity.  BGMs have been widely applied for a wide variety of estimation and detection problems in signal processing \cite{loeliger2004signal,loeliger2007factor}, with applications that include sensor networks \cite{Swami07}, surveillance \cite{ferri2017cooperative}, and  information-seeking control \cite{Meyer15}. They also naturally unify several standard methods from statistical estimation theory, such as the forward-backward algorithm \cite{rabiner1989tutorial}, the Kalman filter \cite{kschischang2001factor}, the particle filter \cite{ihler2005nonparametric}, and the Viterbi algorithm \cite{kschischang2001factor}. 

Beyond learning, estimation and detection, BGMs have also found applications in stochastic control problems, which involve not only estimation of the state of a system, but also  determination of suitable control actions in the presence of uncertainty.  Applications include control of vehicles, robots, factories, or teams of agents. Often, the control problem and inference/estimation problem are considered separate, whereby the controller assumes an estimate of the state and the inference occurs without knowledge of future control. Such a separation principle is only valid in certain cases, such as the celebrated linear quadratic Gaussian (LQG) control \cite{caines2018linear}. 
Over the past decades, several approaches have been proposed to unify inference and control 
\cite{karny_towards_1996,ziebart_maximum_2008,kappen_optimal_2012,watson_stochastic_2019,levine_reinforcement_2018,toussaint_robot_2009,hoffmann_linear_2017}, largely based on BGMs. 
The core ideas of these approaches can already be found in the early work \cite{karny_towards_1996}, which posed the role of a controller as follows: \emph{``A controller of a  stochastic system `shapes' the joint pdf describing the closed-loop behaviour. The `optimal' controller should make this joint pdf as close as possible to a desired pdf.''} With this in mind, \cite{karny_towards_1996} poses an ideal state distribution and control distribution, after which an optimized controller can be found be minimizing a Kullback-Leibler divergence (KLD). In terms of mathematical tractability, an important improvement was the use of Bayesian graphical models \cite{loeliger2007factor}, which led to the methods in \cite{kappen_optimal_2012,watson_stochastic_2019,hoffmann_linear_2017}. In \cite{kappen_optimal_2012}, a similar idea as \cite{karny_towards_1996} was proposed, which allowed the formulation of control cost as a KLD, which could be solved by approximate inference on the corresponding graphical model.
In \cite{watson_stochastic_2019}, a linear transformation of the control cost function is replaced by a log-likelihood function and an optimized controller is found by an expectation-maximization procedure over the corresponding factor graph. A similar idea was introduced in \cite{levine_reinforcement_2018,toussaint_robot_2009} where an artificial observation and associated likelihood was introduced so that state trajectories with highest posterior probability also have lowest associated control cost. In \cite{watson_stochastic_2019,toussaint_robot_2009} controllers similar to LQG were found. In \cite{hoffmann_linear_2017} the LQG control problem was targeted specifically, and under perfect knowledge of the current state, the exact LQG controller was recovered by an inference-based controller. It should be noted that some of the above works aim to find a policy (i.e. a mapping from state estimate to control) while others aim to determine an optimal control sequence. 

More generally, stochastic optimal control problems have been solved using a diverse range of approaches, where \emph{model-predictive control (MPC)} \cite{lee_model_2011} and \emph{reinforcement learning (RL)} \cite{sutton_reinforcement_2018} are arguably the most prominent approaches. When a model of the system is available, the control problem becomes a Markov decision process, which can, in principle, be solved through dynamic programming \cite{puterman_markov_2014}. If no model is available, RL can provide model-free solutions that learn state-action mappings from interactions with the system \cite{degris_model-free_2012}. Recently, there has been work combining these two approaches, originating either from the control theory community \cite{williams_information_2017} or the computer science community \cite{kamthe_data-efficient_2017}.

In addition to MPC and RL, a third and more recent path is the \emph{free energy principle (FEP)}, which originates from cognitive neuroscience as a way to explain biological behavior \cite{friston_free_2006,ramstead_answering_2018}. The main hypothesis is that agents (i) internalize a generative model (GM) of the system, and (ii) perceive and act in such a way as to minimize a free-energy bound on surprise relative to the GM. Interestingly, free-energy minimization is a concept that is also used in RL to encourage exploration and model building \cite{sallans_using_2001}.
Objective functions for any kind of system or application can be included in the GM in the form of a goal prior, which results in formulations of active inference (ActInf) \cite{friston_free-energy_2010}. Despite a large number of publications in the ActInf field including applications in robotics \cite{OliverLanillosheng_2019} and synthesis with reinforcement learning \cite{catal_bayesian_2019}, there have been only few efforts, e.g.,  \cite{hoffmann_linear_2017,ueltzhoffer_deep_2018,schwobel_active_2018,baltieri_active_2019,millidge2020relationship,imohiosen2020active}, to apply it to more traditional stochastic control settings, such as linear quadratic Gaussian (LQG) control.

In this paper, our main aim is to reveal the connections between inference and control over BGM from an ActInf perspective. Specifically, we have the following contributions:
\begin{enumerate}
    \item We propose an ActInf joint inference and control formulation that casts the control problem as an inference problem and explicitly encodes the control cost in the FEP framework;
    \item We show under which conditions the ActInf-based joint inference and control method yields the solution to the original stochastic optimal control problem; 
    \item We prove that LQG can be expressed as a special case of the proposed ActInf joint inference and control method.
\end{enumerate}

The article is structured as follows: In the remainder of this section, we provide an overview of the notation.  Sec.~\ref{sec:problem_statement} introduces the model and optimal control objective. Sec.~\ref{sec:active_inference} introduces the ActInf objective and further notation related to probabilistic model formulations. Sec.~\ref{sec:from_actinf_to_oc} formally relates the objective function of ActInf with stochastic optimal control. Sec.~\ref{sec:relationship} then applies these formulations to a LQG control problem, which is illustrated by the numerical results in Sec.~\ref{sec:simulation}. We conclude with Sec.~\ref{sec:conclusions}.

\subsection*{Notation}
We write a sequence of variables as $s_{t_1:t_2}= \{s_{t_1},\ldots, s_{t_2}\}$. At any current time $t$, we consider a sequence of states, observations and controls as $x=x_{0:t+T}$, $y =y_{1:t+T}$, $u=u_{0:t+T-1}$, with $x_k \in \mathbb{R}^{n_x}$, controls $u_k \in \mathbb{R}^{n_u}$, and observations $y_k \in \mathbb{R}^{n_y}$ respectively, with a time horizon of $T$ time-steps into the future. Note that the start and the end points between the state, observation and control sequence differ slightly. When explicitly required, we denote the realizations of the random variables, such as  (past) observed values, estimates and performed actions, by a bold script.

In order to easily distinguish between the past and future variables, we adopt the following convention: we divide the observations $y$ into past (including present) variables $\underline{y}_t = y_{1:t}$ and future variables $\overline{y}_t = y_{t+1:t+T}$. Similarly, the state sequence $x$ consists of $\underline{x}_t = x_{0:t}$ and $\overline{x}_t = x_{t+1:t+T}$. The control sequence $u$ consists of $\underline{u}_t = u_{0:t-1}$ and $\overline{u}_t = u_{t:t+T-1}$ (with present control included). For notational convenience, we drop the dependence on the current time $t$. For instance, we use $\overline{x}$ instead of $\overline{x}_t$. We use $s_{\setminus t}$ to denote the sequence obtained by removing $s_t$ from $s$. 

Similar to the notation for the sequences, some of the probability density functions (pdfs) are expressed using the notation $\underline{p}(.)$ and $\overline{p}_t(.)$ to emphasize functions of past and future variables, respectively. As usual, marginal and conditional pdfs associated with a given joint pdf are denoted using the same letter/subscript. For instance, the marginal obtained by marginalizing (i.e., integrating) $p_a(s_1,s_2)$ over $s_2$ is denoted by $p_a(s_1) = \int p_a(s_1,s_2) \d{s_2}$. To avoid clutter, we drop the distribution arguments (i.e., we write $p_a$ instead of $p_a(s_1,s_2)$) whenever these dependencies are clear from the context. 

\begin{table}
\caption{Common notations for distributions and functions.}
\label{table:notation:distribution}
\begin{centering}
\begin{tabular}{|l|l|l|}
\hline 
 & \emph{Short Description} & \emph{Normalized}\tabularnewline
\hline 
\hline 
$p$ & joint distribution & {yes}\tabularnewline
\hline 
$p_t$ & generative model at time $t$ & yes\tabularnewline
\hline 
$\tilde{p}$ & goal priors & {yes}\tabularnewline
\hline 
$f_{t}$ & goal-constrained generative model & no\tabularnewline
\hline 
$p_{p}$ & posterior distribution of hidden variables & yes\tabularnewline
\hline 
$q$ & belief / variational posterior & yes\tabularnewline
\hline 
$\pi_t$ & stochastic policy mapping & yes\tabularnewline
\hline 
\end{tabular}
\par\end{centering}
\end{table}

\section{System Model}
\label{sec:problem_statement}

\subsection{Dynamical System Model}
We consider the following dynamical system with the state-space model (SSM):
\begin{subequations}
\label{eq:process}
\begin{align}
    x_{t+1} & \sim p(x_{t+1} |x_t, u_t), \quad t \geq 0\,,\\
    y_t & \sim p(y_t |x_t), \quad t \geq 1\,,
\end{align}
\end{subequations}
where $x_0 \sim p(x_0)$ and $u_0=0$. Using the system definition in \eqref{eq:process}, the probabilistic system model for the state and outcome sequence for a given control sequence over a time window of $k \in [0, t+T]$ can be expressed as follows
\begin{align}
    p(y, x| u) = p(x_0)\!\!\!\! \prod_{k=0}^{t+T-1} \!\!\!\! p(y_{k+1} |x_{k+1}) p(x_{k+1} |x_k, u_k). \label{eq:system_model}
\end{align}
At time $t$,  we have the probabilistic system model
\begin{align}
    p_t(y, x| u) = \frac{p(\underline{y}_t = \underline{\bm{y}}_t, \overline{y}_t,  x| \underline{u}_t = \underline{\bm{u}}_t, \overline{u}_t)}{\int p(\underline{y}_t = \underline{\bm{y}}_t, \overline{y}_t,  x| \underline{u}_t = \underline{\bm{u}}_t, \overline{u}_t) \d{\overline{y}_t} \d{x}},\label{eq:system_model_t}
\end{align}
where  $\underline{y}_t$ and $\underline{u}_t$ are set to their realizations ($\underline{\bm{y}}_t$ and $\underline{\bm{u}}_t$). We will generally omit the explicit dependence on $\underline{\bm{y}}_t$ and $\underline{\bm{u}}_t$ and instead rely on the sub-script $t$  to indicate that past controls and observations are fixed in $p_t$. Since $p_t$ is obtained by plugging in the values of the realizations, we re-normalize. An overview of the common distributions used in this paper together with their normalization status is provided in Table~\ref{table:notation:distribution}.

\subsection{Control Objective}
\label{sec:control_objective}
We consider stochastic policy mappings in the form of $\pi_k(u_k | y_{1:t})$ from the set of measurements (up to the current time $t$) to the control at time $k$ where $k \in [t, t+T]$. The objective is to find the mappings $\pi_k$ that minimize the expected cost $\mathcal{J}_t$ over current and future states $x_k$ and controls $u_k$, defined as:
\begin{align}
    \mathcal{J}_t = \sum_{k=t}^{t+T} \mathbb{E}_{p_t, \pi_k}\left[\ell_k(x_k, u_k) \right]\,, \label{eq:cost}
\end{align}
where $p_t$ is the probabilistic system model as expressed in \eqref{eq:system_model_t}, and $\ell_k(x_k, u_k)\ge 0$ is the cost function at time-step $k$ that encodes the cost of being in state $x_k$ and applying the control $u_k$. The realization for the current control (action) is then determined using the stochastic policy mapping $\pi_t$. 

In particular, the control is $u_t=\bm{u}_{t,\pi}^*$, where
\begin{align}\label{eq:standard:action}
\bm{u}_{t,\pi}^* = g(\pi_t^*)
\end{align}
with 
\begin{align}
\pi_t^* = \arg \min_{\pi_t} \mathcal{J}_t
\end{align}
and where $g(\cdot)$ represents the mapping from the probability distribution to a single action $\bm{u}_t$, which can be chosen, for instance, as the mean or the mode of $\pi_t^*$ or as a sample (i.e. realization) from $\pi_t^*$ \cite{schwobel_active_2018,b_GladLjung,attias_planning_2003,sutton_reinforcement_2018}. This article considers a sliding horizon, i.e., after taking the action at the current time instant $t$ and obtaining the next observation, the stochastic policies are again determined by looking $T$ steps ahead. 

\emph{Example:} A typical cost function is the quadratic cost:
\begin{align}
    \ell_k(x_k, u_k) = \ell(x_k, u_k) = \tfrac{1}{2}x^{\T}_k Q x_k + \tfrac{1}{2}u^{\T}_k R u_k\,, \label{eq:quadratic_cost}
\end{align}
for  $Q \in \mathbb{R}^{n_x \times n_x}, Q \succeq 0$  and $R \in \mathbb{R}^{n_u \times n_u}, R \succ 0$.

\section{Active Inference}
\label{sec:active_inference}
In this section, we describe the ActInf approach and the FEP. The concepts and approaches in this section have similarities to the control as inference literature \cite{kappen_optimal_2012,watson_stochastic_2019,hoffmann_linear_2017}, but are here presented from the ActInf perspective. The main idea of ActInf is that at each time $t$, the controller minimizes the free energy functional  $\mathcal{F}_t[q]$, defined as\cite{friston_free_2006} $\mathcal{F}_t[q] = \KL[q||f_t]$, where  $\KL[q||f_t]$ is the Kullback-Leibler divergence, $q$ represents a variational distribution (the optimization variable) and $f_t$ represents the known generative model $p_t$ with substituted observations or with modifications with more general constraints. 
Each of these concepts will now be explained in detail. 
Minimization of the free energy, and the well-established framework of minimization of Bayesian surprise are closely connected.  We further discuss this relationship in Remark~\ref{rem:FEP_Suprise}. 

\subsection{Generative Model and Goal Priors}

\subsubsection{The Generative Model}
The notation introduced earlier allows us to concisely write the system model at time $t$ (\ref{eq:system_model_t}) in terms of a past and a future contribution:
\begin{subequations}
\label{eq:so}
\begin{align}
    p_t(y, x| u) &= p_t(\yf,\yp,\xf,\xp|\uf,\up ),\\
    &=p_t(\yp,\xp| \uf,\up) \, p_t(\yf,\xf | \yp, \xp,\uf,\up ),\\
    &= \underline{p}_t(\underline{y}, \underline{x} | \underline{u})\, \overline{p}_t(\overline{y}, \overline{x} | x_t, \overline{u})\,. \label{eq:so_past_future}
\end{align}
\end{subequations}
We note that the first factor depends on past controls and the second on the future controls. Both factors condition on the controls, and $p_t$ does not incorporate the control cost $\mathcal{J}_t$. 

\subsubsection{Goal Priors}
\label{sec:goal_priors}
The designer of the agent should govern the system behaviour towards desirable system states, e.g. the exit of a maze. In order to achieve this, ActInf introduces the concept of a \emph{prior belief on the future outcomes}
\cite{de_vries_factor_2017,parr_generalised_2018,van_de_laar_simulating_2019} (or equivalently referred as a goal prior) which constrains the system model \eqref{eq:system_model}. This goal prior is set by the designer of the agent and encodes the future states that are desirable, in other words the states that we want to be \emph{unsurprising}  for the agent. Actions selected as a result of free energy minimization will then move the agent as close as possible to these unsurprising (desired) states.     
A goal prior is added as an additional factor to the system model description \cite{parr_generalised_2018,van_de_laar_simulating_2019}, leading to the goal-constrained (unnormalized) GM:
\begin{align}
    f_t(y, x, \overline{u} | \underline{u}) = \underbrace{p_t(y, x| u)}_{\substack{\text{generative}\\\text{model}}}\, \underbrace{\tilde{p}(\overline{y}, x_t, \overline{x}, \overline{u})}_{\text{goal prior}}\,. \label{eq:gm_goals}
\end{align}

In order to relate the goal prior to the traditional control cost, a natural choice is:
\begin{align}
    \tilde{p}(\overline{y}, x_t, \overline{x}, \overline{u}) &= \frac{1}{\Gamma} \exp\left(-\lambda \sum_{k=t}^{t+T} \ell_k(x_k, u_k)\right), \label{eq:goal_prior}
\end{align}
where $\Gamma$ the the normalization constant and $\lambda\ge 0$ is the scaling factor. For the quadratic cost of \eqref{eq:quadratic_cost} the goal prior factors into independent Gaussians (i.e., consists of factors in the form $\propto \exp(-\lambda \tfrac{1}{2}x^{\T}_k Q x_k) \exp(-\lambda \tfrac{1}{2}u^{\T}_k R u_k)$), where the weighting matrices $Q$ and $R$ take the role of (scaled) precisions. A related probabilistic approach is described in \cite{toussaint_robot_2009}, where a binary reward is defined by using a cost function.

\subsection{Free Energy Objective}
\label{sec:free_energy_objective}

Consider the latent (hidden) variables at time $t$: $\yf, x, \uf$. 
Note that the state sequence $x$ is unknown for both the future and the past, whereas for the observations and the controls, only the future variables are unknown. Let us consider a variational posterior distribution $q(\yf, x, \uf|\yp,\up)$ defined over the latent variables. Here, the label {\it{variational}} refers to the fact that the objective function \eqref{eq:F_t} is optimized by {\it{variations}} in the conditional \cite{blei_variational_2017}. Note that $q(\yf, x, \uf|\yp,\up)$ is a posterior conditioned on the past observations and controls.  To avoid notational clutter, we adopt a mainstream notational convention in probabilistic inference where conditioning is dropped from the variational posterior distribution,  and represent  $q(\yf, x, \uf|\yp,\up)$ as $q(\yf, x, \uf)$ or simply as $q$.

The free energy $\mathcal{F}_t[q]$ is defined as follows \cite{friston_free_2006}:
\begin{align}
    \mathcal{F}_t[q] &= \KL[q||f_t]\,, \label{eq:F_t}
\end{align}
where  $\KL[q||f_t] \triangleq \int q(s) \log\!\left({q(s)}/{f_t(s)}\right)\d{s}$ is the Kullback-Leibler (KL) divergence (i.e., relative entropy). The KL divergence is an information-theoretic concept that quantifies how much one probability distribution differs from another distribution \cite{b_coverThomas}.
\begin{figure*}
    \centering
    \resizebox{\textwidth}{!}{

\begin{tikzpicture}
    [node distance=15mm,auto,>=stealth']


    \node (x_t_min) {$\dots$};
    \node[box, right of=x_t_min, node distance=18mm] (trans_t) {};
    \node[clamped, above of=trans_t] (u_t_min) {};
    \node[detbox, right of=trans_t] (eq_t) {$=$};
    \node[box, below of=eq_t, node distance=12mm] (obs_t) {};
    \node[clamped, below of=obs_t] (y_t) {};
    \node[detbox, right of=eq_t] (eq_tt) {$=$};
    \node[box, above of=eq_tt, node distance=17mm] (q_t) {};

    \node[box, right of=eq_tt] (trans_t_p1) {};
    \node[box, above of=trans_t_p1, node distance=17mm] (u_t) {};
    \node[detbox, right of=trans_t_p1] (eq_t_p1) {$=$};
    \node[box, below of=eq_t_p1, node distance=12mm] (obs_t_p1) {};
    \node[box, below of=obs_t_p1, node distance=17mm] (y_t_p1) {};    
    \node[right of=eq_t_p1] (dots) {$\dots$};

    \node[detbox, right of=dots, node distance=20mm] (eq_tt_pT) {$=$};
    \node[box, above of=eq_tt_pT, node distance=17mm] (q_t_pT) {};
    \node[box, right of=eq_tt_pT] (trans_t_pT) {};
    \node[box, above of=trans_t_pT, node distance=17mm] (u_T) {};
    \node[detbox, right of=trans_t_pT] (eq_t_pT) {$=$};
    \node[box, below of=eq_t_pT, node distance=12mm] (obs_t_pT) {};
    \node[box, below of=obs_t_pT, node distance=17mm] (y_t_pT) {};    
    \node[box, above of=eq_t_pT, node distance=17mm] (q_T) {};

    \draw (x_t_min) -- node[anchor=south]{$x_{t-1}$} (trans_t);
    \draw (u_t_min) -- node[anchor=west]{$u_{t-1}$} (trans_t);
    \draw (trans_t) -- (eq_t);
    \draw (eq_t) -- (obs_t);
    \draw (obs_t) -- node[anchor=west]{$y_t$} (y_t);
    \draw (eq_t) -- node[anchor=south]{$x_t$} (eq_tt);
    \draw (q_t) -- (eq_tt);

    \draw (eq_tt) -- (trans_t_p1);
    \draw (u_t) -- node[anchor=west]{$u_t$} (trans_t_p1);
    \draw (trans_t_p1) -- (eq_t_p1);
    \draw (eq_t_p1) -- (obs_t_p1);
    \draw (obs_t_p1) -- node[anchor=west]{$y_{t+1}$} (y_t_p1);
    \draw (eq_t_p1) -- node[anchor=south]{$x_{t+1}$} (dots);

    \draw (dots) -- node[anchor=south]{$x_{t+T-1}$} (eq_tt_pT);
    \draw (q_t_pT) -- (eq_tt_pT);
    \draw (eq_tt_pT) -- (trans_t_pT);
    \draw (u_T) -- node[anchor=west]{$u_{t+T-1}$} (trans_t_pT);
    \draw (trans_t_pT) -- (eq_t_pT);
    \draw (eq_t_pT) -- (obs_t_pT);
    \draw (obs_t_pT) -- node[anchor=west]{$y_{t+T}$} (y_t_pT);
    \draw (q_T) -- node[anchor=west]{$x_{t+T}$} (eq_t_pT);

    \draw[dashed] (-0.6,2.3) rectangle (3.8,-3.5);
    \draw[dashed] (4.3,0.6) rectangle (14.9,-1.8);
    \draw[dashed] (4.3,2.3) -- (15.9,2.3) -- (15.9,-3.5) -- (4.3,-3.5) -- (4.3,-2.3) -- (15.4,-2.3) -- (15.4,1.1) -- (4.3,1.1) -- (4.3,2.3);

    \node (present) at (1,-3.0) {$\underline{p}(\underline{y}, \underline{x} | \underline{u})$};
    \node (future) at (11.5,-1.4) {$\overline{p}(\overline{y}, \overline{x} | x_t, \overline{u})$};
    \node (goals) at (9,1.5) {$\tilde{p}(\overline{y}, x_t, \overline{x}, \overline{u})$};
\end{tikzpicture}}
    \caption{Forney-style factor graph representation of the goal-constrained generative model \eqref{eq:gm_goals} with indicated factorizations. Observations are indicated by small solid nodes.}
    \label{fig:gm_ffg}
\end{figure*}
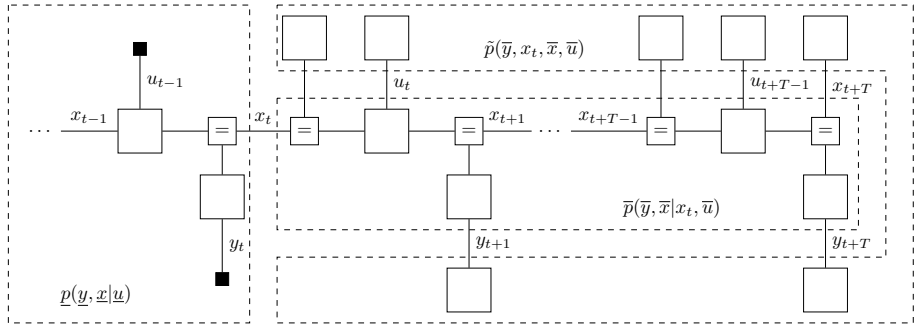
\begin{figure*}
    \centering
    \resizebox{\textwidth}{!}{

\begin{tikzpicture}
    [node distance=15mm]


    \node (x_t_min) {$\dots$};
    \node[box, right of=x_t_min, node distance=18mm] (trans_t) {};
    \node[clamped, above of=trans_t] (u_t_min) {};
    \node[detbox, right of=trans_t] (eq_t) {$=$};
    \node[box, below of=eq_t, node distance=12mm] (obs_t) {};
    \node[clamped, below of=obs_t] (y_t) {};
    \node[detbox, right of=eq_t, node distance=20mm] (eq_tt) {$=$};
    \node[box, above of=eq_tt] (q_t) {};

    \node[box, right of=eq_tt] (trans_t_p1) {};
    \node[box, above of=trans_t_p1] (u_t) {};
    \node[detbox, right of=trans_t_p1] (eq_t_p1) {$=$};
    \node[box, below of=eq_t_p1, node distance=12mm] (obs_t_p1) {};
    \node[box, below of=obs_t_p1] (y_t_p1) {};    
    \node[right of=eq_t_p1] (dots) {$\dots$};

    \node[detbox, right of=dots, node distance=20mm] (eq_tt_pT) {$=$};
    \node[box, above of=eq_tt_pT] (q_t_pT) {};
    \node[box, right of=eq_tt_pT] (trans_t_pT) {};
    \node[box, above of=trans_t_pT] (u_T) {};
    \node[detbox, right of=trans_t_pT] (eq_t_pT) {$=$};
    \node[box, below of=eq_t_pT, node distance=12mm] (obs_t_pT) {};
    \node[box, below of=obs_t_pT] (y_t_pT) {};    
    \node[box, above of=eq_t_pT] (q_T) {};

    \draw (x_t_min) -- node[anchor=north]{$x_{t-1}$} node[anchor=south]{$\rightarrow$} (trans_t);
    \draw (u_t_min) -- node[anchor=west]{$u_{t-1}$} node[anchor=east]{$\downarrow$} (trans_t);
    \draw (trans_t) -- node[anchor=south]{$\rightarrow$} (eq_t);
    \draw (eq_t) -- node[anchor=east]{$\uparrow$} (obs_t);
    \draw (obs_t) -- node[anchor=west]{$y_t$} node[anchor=east]{$\uparrow$} (y_t);
    \draw (eq_t) -- node[anchor=north]{$x_t$} (eq_tt);
    \draw (q_t) -- node[anchor=east]{$\downarrow$} (eq_tt);

    \draw (eq_tt) -- node[anchor=south]{$\rightarrow$} (trans_t_p1);
    \draw (u_t) -- node[anchor=west,pos=0.8]{$u_t$} (trans_t_p1);
    \draw (trans_t_p1) -- node[anchor=north]{$\leftarrow$} (eq_t_p1);
    \draw (eq_t_p1) -- node[anchor=east]{$\uparrow$} (obs_t_p1);
    \draw (obs_t_p1) -- node[anchor=west]{$y_{t+1}$} node[anchor=east]{$\uparrow$} (y_t_p1);
    \draw (eq_t_p1) -- node[anchor=south]{$x_{t+1}$} (dots);

    \draw (dots) -- node[anchor=south]{$x_{t+T-1}$} (eq_tt_pT);
    \draw (q_t_pT) -- node[anchor=east]{$\downarrow$} (eq_tt_pT);
    \draw (eq_tt_pT) -- node[anchor=north]{$\leftarrow$} (trans_t_pT);
    \draw (u_T) -- node[anchor=west]{$u_{t+T-1}$} node[anchor=east]{$\downarrow$} (trans_t_pT);
    \draw (trans_t_pT) -- node[anchor=north]{$\leftarrow$} (eq_t_pT);
    \draw (eq_t_pT) -- node[anchor=east]{$\uparrow$} (obs_t_pT);
    \draw (obs_t_pT) -- node[anchor=west]{$y_{t+T}$} node[anchor=east]{$\uparrow$} (y_t_pT);
    \draw (q_T) -- node[anchor=west]{$x_{t+T}$} node[anchor=east]{$\downarrow$} (eq_t_pT);

    \draw[dashed] (-0.5,1.9) rectangle (3.9,-3.1);

    \msg{up}{right}{eq_t}{eq_tt}{0.5}{A};
    \msg{left}{down}{u_t}{trans_t_p1}{0.42}{E};
    \msg{right}{up}{u_t}{trans_t_p1}{0.42}{D};
    \msg{down}{left}{dots}{eq_tt_pT}{0.5}{B};
    \msg{down}{left}{eq_t_p1}{dots}{0.5}{C};
\end{tikzpicture}}
    \caption{Forney-style factor graph specification of the inference algorithm on the goal-constrained generative model. Here, message \protect\smallcircled{A} represents a state estimate that summarizes past control and observations. The product of messages \protect\smallcircled{D} and \protect\smallcircled{E} yields a posterior belief over the current control, the mode of which is taken as the present action.}
    \label{fig:inference_ffg}
\end{figure*}
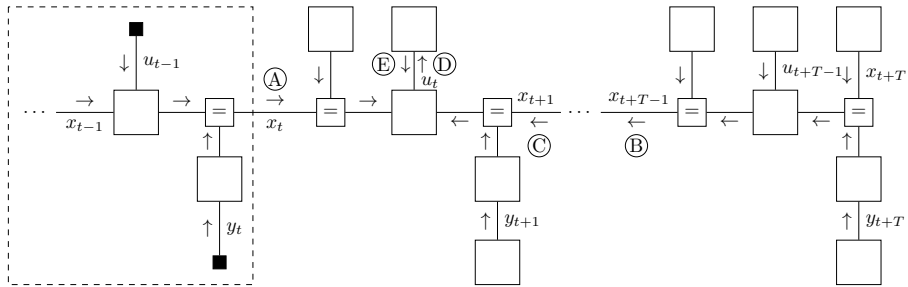
By straightforward manipulation, the free energy can be decomposed as follows: 
\begin{align}
    \mathcal{F}_t[q] &= \underbrace{-\log Z}_{\text{surprise}} + \underbrace{\mathbb{E}_q\!\left[\log \frac{q(\overline{y}, x, \overline{u})}{p_p(\overline{y}, x, \overline{u} | \underline{y}, \underline{u})} \right]}_{\text{posterior divergence}}\,, \label{eq:surprise}\\
    &={-\log Z+ \KL[q||p_p],
    \label{eq:surprise:KL}
    }
\end{align}
where $p_p$ denoted the exact (Bayesian) posterior, and where $Z = \int f_t(y, x, \overline{u} | \underline{u}) \d{\overline{y}} \d{x} \d{\overline{u}}$, with substituted past observations $\underline{y}$ and controls $\underline{u}$; and  $p_p(\overline{y}, x, \overline{u} | \underline{y}, \underline{u}) = \frac{1}{Z} f_t(y, x, \overline{u} | \underline{u}) $. 
Since the posterior (KL) divergence term is always positive, the free energy provides an upper bound on the exact (Bayesian) surprise. This decomposition is often employed to justify the use of free energy as a tool for (approximate) inference and model selection \cite{attias2000variational}. 

\begin{rem}[Relation between FEP and surprise] \label{rem:FEP_Suprise}
	Generally, the generative model $p$ is a pdf over hidden states (say $x$) and observations (say $y$), while $q$ is a pdf only over hidden states. Hence, $p(x,y)=p(y|x)p(x)$, in which $p(x)$ represents a prior. Substituting an observation $y = \bm{y}$ in the model, yields $f_t(x) = p(y=\bm{y} | x)p(x)$, which represents the product of a likelihood and the prior. We are interested in obtaining a posterior belief $p_p(x) = f_t(x)/Z$. However, the normalizing constant (Bayesian evidence) $Z=\int f_t(x)\d{x}$ is often intractable to compute, because it involves an integral over all hidden state configurations. As a result, it is often prohibitively expensive (in terms of computational power) to obtain an exact solution for the posterior $p_p$.
	Instead, posterior inference is often cast as a free energy optimization problem, where the free energy factorizes as $\mathcal{F}_t[q] = -\log Z + \KL[q||p_p]$ as in \eqref{eq:surprise}--\eqref{eq:surprise:KL}. Minimization of $\mathcal{F}_t$ thus maximizes Bayesian evidence, while minimizing posterior divergence, making $q$ a close approximation to the (usually intractable) posterior $p_p$.
\end{rem}

\subsection{Control}

At time $t$, the objective is to find the $q$ that minimizes the free energy, i.e.,
\begin{align}
    q^*_t = \arg \min_{q}\,\mathcal{F}_t[q].
\end{align}
Looking at  \eqref{eq:gm_goals} and \eqref{eq:F_t}, we observe that after optimization, the variational distribution $q^*_t$ simultaneously accounts for the constraints enforced by the system model \eqref{eq:system_model} and the goal prior \eqref{eq:goal_prior}. The posterior for the current control is obtained by marginalization, i.e., $q_t^*(u_t) = \idotsint q^*_t(\overline{y}, x, \overline{u}) \d{\overline{y}}, \d{x}, \d{\overline{u}_{\setminus t}}$, where $\overline{u}_{\setminus t}$ denotes the sequence obtained by removing $u_t$ from $\overline{u}$.
The current action is then chosen as
\begin{align}\label{eq:actinf:action}
\bm{u}_{t,q}^*= g(q_t^*)\,,
\end{align}
where $g(.)$ is the same as in \eqref{eq:standard:action}.

\subsection{Free Energy Minimization by Message Passing on a Forney-style factor graph}
\label{sec:fe_by_mp}
It is instructive to separate inference relating to the past/present from the inference relating to the future. To this end, we substitute the GM $f_t(y, x, \overline{u} | \underline{u})$ of \eqref{eq:gm_goals} into \eqref{eq:F_t} and use \eqref{eq:so_past_future} to factorize the free energy in the following form with separate contributions from the present ($\mathcal{V}_t[q]$) and (expected) future ($\mathcal{G}_t[q]$):
\begin{align}
    \mathcal{F}_t[q] &= \underbrace{\mathbb{E}_q\!\left[\log \frac{{q}(\underline{x})}{\underline{p}_t(\underline{y}=\underline{\bm{y}}, \underline{x}| \underline{u}=\underline{\bm{u}})}\right]}_{\mathcal{V}_t[q]} + \underbrace{\mathbb{E}_q\!\left[\log \frac{{q}(\overline{y}, \overline{x}, \overline{u} | \underline{x})}{\overline{p}_t(\overline{y}, \overline{x} | x_t, \overline{u})\,\tilde{p}(\overline{y}, x_t, \overline{x}, \overline{u})}\right]}_{\mathcal{G}_t[q]}\,. \label{eq:FVG}
\end{align}

In practice, the optimization of $\mathcal{F}_t$ is often intractable and a specific choice for the factorization of $q$ is made to aid computation \cite{bishop2006pattern}. In our case, due to the factorization of the GM, we can optimize $\mathcal{F}_t$ exactly by belief propagation over the factor graph of the GM \cite{GBP2005,heskes2003stable}. To minimize $\mathcal{V}_t$ and $\mathcal{G}_t$, a Forney-style factor graph (FFG) offers a convenient visual representation of a factorized function \cite{forney_codes_2001}, and is especially well-suited for representing probabilistic models \cite{loeliger_introduction_2004}. In an FFG, edges represent variables and nodes (factors) represent relations between variables. The FFG representation of the GM \eqref{eq:gm_goals} with substituted factorizations \eqref{eq:system_model} and included goal priors \eqref{eq:goal_prior} is drawn in Fig.~\ref{fig:gm_ffg}. The free energy \eqref{eq:F_t} is then minimized by message passing \cite{heskes2003stable,GBP2005,dauwels_variational_2007,van_de_laar_simulating_2019} on the FFG representation of the goal-constrained GM. Message passing can be interpreted as first minimizing $\mathcal{V}_t$ and then minimizing a modified version of $\mathcal{G}_t$, based on the outcome of minimizing $\mathcal{V}_t$ \cite{loeliger2007factor}.

\subsubsection*{Minimizing $\mathcal{V}_t[q]$} Message passing  yields a message $\smallcircled{A}$ (Fig.~\ref{fig:inference_ffg}) that represents the current state estimate, given past observations and controls. This message summarizes the information contained within the dashed box:
\begin{align}
    \mu_{\smallcircled{A}}(x_t) &\triangleq \int \underline{p}_t(\underline{y}=\underline{\bm{y}}, \underline{x} | \underline{u}=\underline{\bm{u}}) \d{x_{0:t-1}} \label{eq:mu_a_integral}\,.
\end{align}
From the perspective of stochastic optimal control, message $\smallcircled{A}$ connects with the estimator. Moreover, for a linear Gaussian state-space model, the recursive message updates for computing $\smallcircled{A}$ constitute a Kalman filter \cite{korl2005factor}. 

\subsubsection*{Minimizing $\mathcal{G}_t[q]$}
In order to minimize $\mathcal{V}_t[q]+\mathcal{G}_t[q]$, we re-normalize the  message $\mu_{\smallcircled{A}}(x_t)$ to obtain a prior $p_e(x_t)$, i.e.,
\begin{align}\label{eq:def:pe}
	p_e(x_t) \triangleq \frac{1}{C_e}  \mu_{\smallcircled{A}}(x_t)
\end{align}
where 
\begin{align}\label{eq:def:ce}
	C_e =\int \underline{p}_t(\underline{y}=\underline{\bm{y}}, \underline{x} | \underline{u}=\underline{\bm{u}}) \d{x_{0:t}}.
\end{align}
Here, the subscript $e$ emphasizes the fact that $p_e$ represents the pdf of an estimate (of the current state). We then define a modified objective for the expected future
\begin{align}
    \tilde{\mathcal{G}}_t[q] =\mathbb{E}_q\!\left[\log \frac{{q}(\overline{y}, \overline{x}, \overline{u} , {x}_t)}{p_e(x_t) \overline{p}_t(\overline{y}, \overline{x} | x_t, \overline{u})\,\tilde{p}(\overline{y}, x_t, \overline{x}, \overline{u})}\right], \label{eq:G_tilde}
\end{align}
which can again be minimized by message passing. This yields messages $\smallcircled{B}$ -- $\smallcircled{E}$ (Fig.~\ref{fig:inference_ffg}) by a backward recursion (smoothing pass) over the GM of future variables. The product of $\smallcircled{D}$ and $\smallcircled{E}$ then leads to a marginal belief $q_t^*(u_t)$.  Then,  the current control action is obtained using \eqref{eq:actinf:action}.

\section{From Active Inference To Stochastic Optimal Control}
\label{sec:from_actinf_to_oc}
The main question we're interested in is the following:  ``\emph{When does \eqref{eq:actinf:action} provide the same control actions as \eqref{eq:standard:action}?}''. In other words, when can the ActInf framework be used to solve the traditional stochastic control problem? Below, we investigate this question. Since the goal priors only appear in $\tilde{\mathcal{G}}_t$ and $\mathcal{V}_t$ can be minimized independently, we focus exclusively on the minimization of $\tilde{\mathcal{G}}_t$. We formulate two conditions under which minimizing $\tilde{\mathcal{G}}_t$ reduces to minimizing the stochastic optimal control objective \eqref{eq:cost}. 

Note that  $\tilde{\mathcal{G}}_t[q]$ can be written as
\begin{align*}
	\tilde{\mathcal{G}}_t[q] =\mathbb{E}_q\!\left[\log \frac{{q}(\overline{y}, \overline{x}, \overline{u} , {x}_t)}{p_e(x_t) \overline{p}_t(\overline{y}, \overline{x} | x_t, \overline{u})}\right]-\mathbb{E}_q\!\left[\log \tilde{p}(\overline{y}, x_t, \overline{x}, \overline{u}) \right]\,.
\end{align*}
Then, minimizing $\tilde{\mathcal{G}}_t[q]$ is equivalent to minimizing $\mathcal{G}^{\dagger}_t[q]$:
\begin{align}
	\mathcal{G}^{\dagger}_t[q] = \mathbb{E}_q\!\left[\log \frac{{q}(\overline{y}, \overline{x}, \overline{u} , {x}_t)}{p_e(x_t) \overline{p}_t(\overline{y}, \overline{x} | x_t, \overline{u})}\right]\! \!+\! \lambda \mathbb{E}_q\!\left[\sum_{k=t}^{t+T}\ell(x_k, u_k)\right]\,, \label{eq:G_t_step}
\end{align}
where we substituted the goal prior from \eqref{eq:goal_prior} and omitted the additive constants that do not depend on the optimization variables.

A striking difference between the optimal control and ActInf objective is that the optimal control objective \eqref{eq:cost} involves an expectation w.r.t. the system model $p_t$ and policy mapping $\pi$, while the free energy involves an expectation w.r.t. the variational distribution $q$. In order to compare the solutions, we need to create an equal footing.

\subsection{Rewriting $\mathcal{G}^{\dagger}_t$}
We start by noting that all arguments of $q(\yf, x, \uf)$ that are not within the expectation brackets in \eqref{eq:G_t_step} are marginalized, i.e.,  $\xp_{\setminus t}$ is marginalized. Therefore, $\tilde{\mathcal{G}}_t$ is effectively only optimized with respect to $q(\yf,\xf,\uf,x_t)$.
We then rewrite the variational distribution in terms of the policy by making use of a region-based approximation \cite{cowell1998introduction,GBP2005}. Note that, for a model that is a tree (which is the case for $f_t$), the region-based approximation is exact. Without loss of generality, we write:
\begin{align}
    & q(\overline{y},\overline{x},\overline{u},x_t)=\frac{\prod_{k=t}^{t+T-1}q(y_{k+1}, x_k, x_{k+1}, u_k)}{\prod_{k=t+1}^{t+T-1}q(x_k)} \label{eq:factors} \\
    &= \underbrace{\left[\prod_{k=t}^{t+T-1}q(u_k)\right]}_{\overline{\pi}(\overline{u})}\, \underbrace{\left[\frac{\prod_{k=t}^{t+T-1}q(y_{k+1},x_{k},x_{k+1}|u_k)}{\prod_{k=t}^{t+T-1}q(x_k)}\right]}_{\phi(\overline{y},\overline{x}|x_t, \overline{u})}\, {q(x_t)}\,, \notag
\end{align}
where we simply applied the Bethe factorization to write the variational distribution in terms of a control posterior $\overline{\pi}$, a system posterior $\phi$, and the current-state posterior ${q}(x_t)$.

We now use \eqref{eq:factors} to rewrite the second term of \eqref{eq:G_t_step}, as:
\begin{subequations}
\begin{align}
    & \lambda \mathbb{E}_q\!\left[\sum_{k=t}^{t+T}\ell(x_k, u_k)\right] =  
     \lambda \mathbb{E}_{q}\!\left[\frac{p_e(x_t)}{p_e(x_t)}\frac{\overline{p}_t(\overline{y}, \overline{x}|x_t, \overline{u})}{\overline{p}_t(\overline{y}, \overline{x}|x_t, \overline{u})}\sum_{k=t}^{t+T}\ell(x_k, u_k)\right] \\
     &= 
     \lambda \mathbb{E}_{\pif}\!\left[ \mathbb{E}_{q(x_t), \phi}\!\left[\frac{p_e(x_t)}{p_e(x_t)}\frac{\overline{p}_t(\overline{y}, \overline{x}|x_t, \overline{u})}{\overline{p}_t(\overline{y}, \overline{x}|x_t, \overline{u})}\sum_{k=t}^{t+T}\ell(x_k, u_k) \right] \right]\,,
     \label{eq:G_t_cost_term:tower} \\
    &= \lambda \mathbb{E}_{p_e, \overline{p}_t, \overline{\pi}}\!\left[\frac{q(x_t)}{p_e(x_t )}\frac{\phi(\overline{y}, \overline{x}|x_t, \overline{u})}{\overline{p}_t(\overline{y}, \overline{x}|x_t, \overline{u})}\sum_{k=t}^{t+T}\ell(x_k, u_k)\right]\,, \label{eq:G_t_cost_term} 
\end{align}
\end{subequations}
where in \eqref{eq:G_t_cost_term:tower} we made the expectations due to different terms of $q(\yf,\xf,\uf,x_t) =  \pif (\uf) \phi(\yf,\xf|x_t) q(x_t)$ from \eqref{eq:factors} explicit and in the last step we interchanged distributions in the expectation subscript with distributions in the numerators, i.e., we used $\E_q\!\left[\frac{p(s)}{p(s)} f(s)\right] = \E_p\!\left[\frac{q(s)}{p(s)} f(s)\right]$ for a function $f$ and probability distributions $q$ and $p$.

We now turn to the first term of \eqref{eq:G_t_step}. Again using the factorization $q(\yf,\xf,\uf,x_t) =  \pif (\uf) \phi(\yf,\xf|x_t) q(x_t)$ from  \eqref{eq:factors}, we rewrite this first term as the sum of the negative policy entropy and a posterior divergence:
\begin{align}
    & \mathbb{E}_q\!\left[\log \frac{{q}(\overline{y}, \overline{x}, \overline{u} , {x}_t)}{p_e(x_t) \overline{p}_t(\overline{y}, \overline{x} | x_t, \overline{u})}\right] = \label{eq:G_t_div_term}\\
    &\quad \mathbb{E}_q\!\left[\log \overline{\pi}(\overline{u}) \right] + D(q(x_t) \Vert p_e(x_t)) +
    \mathbb{E}_q\!\!\left[\log \frac{\phi(\overline{y}, \overline{x}| x_t, \overline{u})}{\overline{p}_t(\overline{y}, \overline{x} | x_t, \overline{u})}\right]\!. \notag
\end{align}

Substituting \eqref{eq:G_t_cost_term} and \eqref{eq:G_t_div_term} in \eqref{eq:G_t_step} then reveals the full expression for the ActInf controller objective:
\begin{align}
	\label{eq:G_t_fact}
   	\mathcal{G}^{\dagger}_t[q] &= \mathbb{E}_q\!\left[\log \overline{\pi}(\overline{u}) \right] + D(q(x_t) \Vert p_e(x_t)) + \mathbb{E}_q\!\!\left[\log \frac{\phi(\overline{y}, \overline{x}| x_t, \overline{u})}{\overline{p}_t(\overline{y}, \overline{x} | x_t, \overline{u})}\right] + \notag\\
    &\quad \lambda \mathbb{E}_{p_e, \overline{p}, \overline{\pi}}\!\left[\frac{q(x_t)}{p_e(x_t )}\frac{\phi(\overline{y}, \overline{x}|x_t, \overline{u})}{\overline{p}_t(\overline{y}, \overline{x}|x_t, \overline{u})}\sum_{k=t}^{t+T}\ell(x_k, u_k)\right]\,.
\end{align}

\begin{rem}[Interpretation of the FEP objective]\label{rem:FEP}
	The FEP objective can be seen as a trade-off between two terms: one pulls $q$ towards $p$ (under uninformative future values for the controls) and another that minimizes the cost $\sum_{k=t}^{t+T}\ell(x_k, u_k)$. Minimization of only the first three terms in \eqref{eq:G_t_fact} (i.e. the case with $\lambda=0$) leads to an undetermined  variational distribution $q^*$. 
	Namely, because the first three terms in \eqref{eq:G_t_fact} directly stem from the first term of  Eqn.~\eqref{eq:G_t_step}, we have ${{q}^*(\overline{y}, \overline{x}, \overline{u} , {x}_t)} = {p_e(x_t) \overline{p}_t(\overline{y}, \overline{x} | x_t, \overline{u})}$.  Then, we have $\mathcal{G}^{\dagger}_t[q^*]=0$ (with $\lambda=0$).

	Conversely, minimization of only the last term in \eqref{eq:G_t_fact} (with $\lambda > 0$) leads to a degenerate variational distribution $q^*$ with mass only at a global minimizer of $ \sum_{k=t}^{t+T}\ell(x_k, u_k)$, 
	because this last term is equal to $\mathbb{E}_q\!\left[\sum_{k=t}^{t+T}\ell(x_k, u_k)\right]$ (recall that the last term in \eqref{eq:G_t_fact} is only a rewritten version of the last term in \eqref{eq:G_t_step}). In particular, while minimizing $\mathbb{E}_q\!\left[\sum_{k=t}^{t+T}\ell(x_k, u_k)\right]$ over $q$, since there are no other constraints on $q$, and $q$ can be directly chosen as a distribution with point mass at a global minimizer of $ \sum_{k=t}^{t+T}\ell(x_k, u_k)$. For instance, with $\ell(.)$ defined as \eqref{eq:quadratic_cost}, $q$ with point mass at $x_k=0$, $u_k=0$ $\forall k$ is a minimizer. 
\end{rem}

We now present two sets of conditions under which the actions chosen by the ActInf agent by \eqref{eq:actinf:action} are the same as the stochastic control actions chosen using \eqref{eq:standard:action}.

\subsection{First Condition: Deterministic Model with Point-Estimate}
Let $\mode(\cdot)$ represent the mode of a distribution, where ties between multiple modes (if any) are resolved by uniform random selection.
\begin{thm}
    \label{thm:deterministic}
    Let (i) $\lambda > 0$, (ii) $\phi = \overline{p}_t$, (iii) $q_e = p_e$, and (iv) $g (\cdot) =\mode(\cdot)$.  Then, 
    $\bm{u}_{t,q}^*=\bm{u}_{t,\pi}^*$.
\end{thm}
\begin{proof}
	See Appendix~\ref{sec:pf:thm:deterministic}.     
\end{proof}

The below result shows that Theorem~\ref{thm:deterministic} implies that the optimal solution is recovered for deterministic systems with a point-estimate.
\begin{cor}
    \label{cor:1}
    The conditions (ii) $\phi = \overline{p}_t$, and (iii) $q_e = p_e$ of Theorem~\ref{thm:deterministic} occur in the case of a deterministic model $\overline{p}_t$ in conjunction with a point estimate for the current state.
\end{cor}
\begin{proof}
    In the case of a deterministic model,  $\overline{p}_t$ of \eqref{eq:process} is constrained   to delta functions; i.e.,  $p(x_{k+1} | x_k, u_k) = \delta(x_{k+1} - f_x(x_k, u_k))$ and $p(y_k | x_k) = \delta(y_k - f_y(x_k))$, for some deterministic functions $f_x(\cdot)$ and $f_y(\cdot)$. A point estimate for the current state (after the minimization of $\mathcal{V}_t$) is chosen as $p_e(x_t) \triangleq \delta(x_t - \hat{\bm{x}}_t)$. Then, any condition other than (ii) $\phi = \overline{p}_t$, (iii) $q_e = p_e$ will lead to infinite divergence in \eqref{eq:G_t_div_term}, and hence in \eqref{eq:G_t_fact}. By contradiction, (ii) $\phi = \overline{p}_t$ and (iii) $q_e = p_e$ are the only viable solutions to the minimization of $\tilde{\mathcal{G}}_t$ under the choice of a deterministic model with a point estimate for the current state.
\end{proof}

\subsection{Second Condition: Vanishing State and Control Cost}
We now consider minimizers of $\eqref{eq:G_t_fact}$ as a function of $\lambda$, and define 
\begin{align}
    r_\lambda(\overline{y},\overline{x},x_t,\overline{u}) \triangleq \frac{q^*(x_t)\phi^*(\overline{y},\overline{x}|x_{t},\overline{u})}{p_e(x_t)\overline{p}_t(\overline{y},\overline{x}|x_{t},\overline{u})}=\frac{q^*(\overline{y},\overline{x},x_t|\overline{u})}{p_t(\overline{y},\overline{x},x_t|\overline{u})}.
\end{align}
Note that the distribution $q$ is an argument of $\eqref{eq:G_t_fact}$. Hence, the optimal $q$ depends on $\lambda$. In light of Remark \ref{rem:FEP}, we see that for $\lambda = 0$,  $r_\lambda(\overline{y},\overline{x},x_t,\overline{u}) = 1$, while for $\lambda >0$, $r_\lambda(\overline{y},\overline{x},x_t,\overline{u}) \neq 1$. 
\begin{thm}
	\label{thm:accurateq}
    Let (i) $\lim_{\lambda \to 0^+} \frac{1}{\lambda}\mathbb{E}_{q}\!\!\left[\log r_\lambda(\overline{y},\overline{x},x_t,\overline{u})\right] = 0$,
    (ii) $\lim_{\lambda \to 0^+} r_\lambda(\overline{y},\overline{x},x_t,\overline{u}) = 1$, $\forall \overline{y},\overline{x},x_t,\overline{u}$, 
    and (iii) $\bm{u}_t = \operatorname{mode} \pi_t$. Then, $\bm{u}_{t,q}^*=\bm{u}_{t,\pi}^*$.
\end{thm}
\begin{proof}
	See Appendix~\ref{sec:pf:thm:accurateq}. 
\end{proof}
Condition (ii) requires that, under a vanishing $\lambda$, the second term of \eqref{eq:G_t_fact} grows to zero faster than $\lambda$ itself. Hence, under (ii), the last term will dominate over the second term, retaining the dependence of $\tilde{\mathcal{G}}_t$ on $\ell$ (see the proof for details). Note that if we outright require $\lambda=0$, all dependence on $\ell$ is immediately lost. Instead, the limit ensures that the influence of the cost $\ell$ is retained.

It is not straightforward to see when the conditions (i) -- (ii) of Theorem~\ref{thm:accurateq} apply. In the subsequent sections, we further discuss the implications of Theorem~\ref{thm:deterministic} and Theorem~\ref{thm:accurateq} for the special case of a linear Gaussian SSM.

\begin{figure*}
    \centering
    \resizebox{\textwidth}{!}{\includegraphics{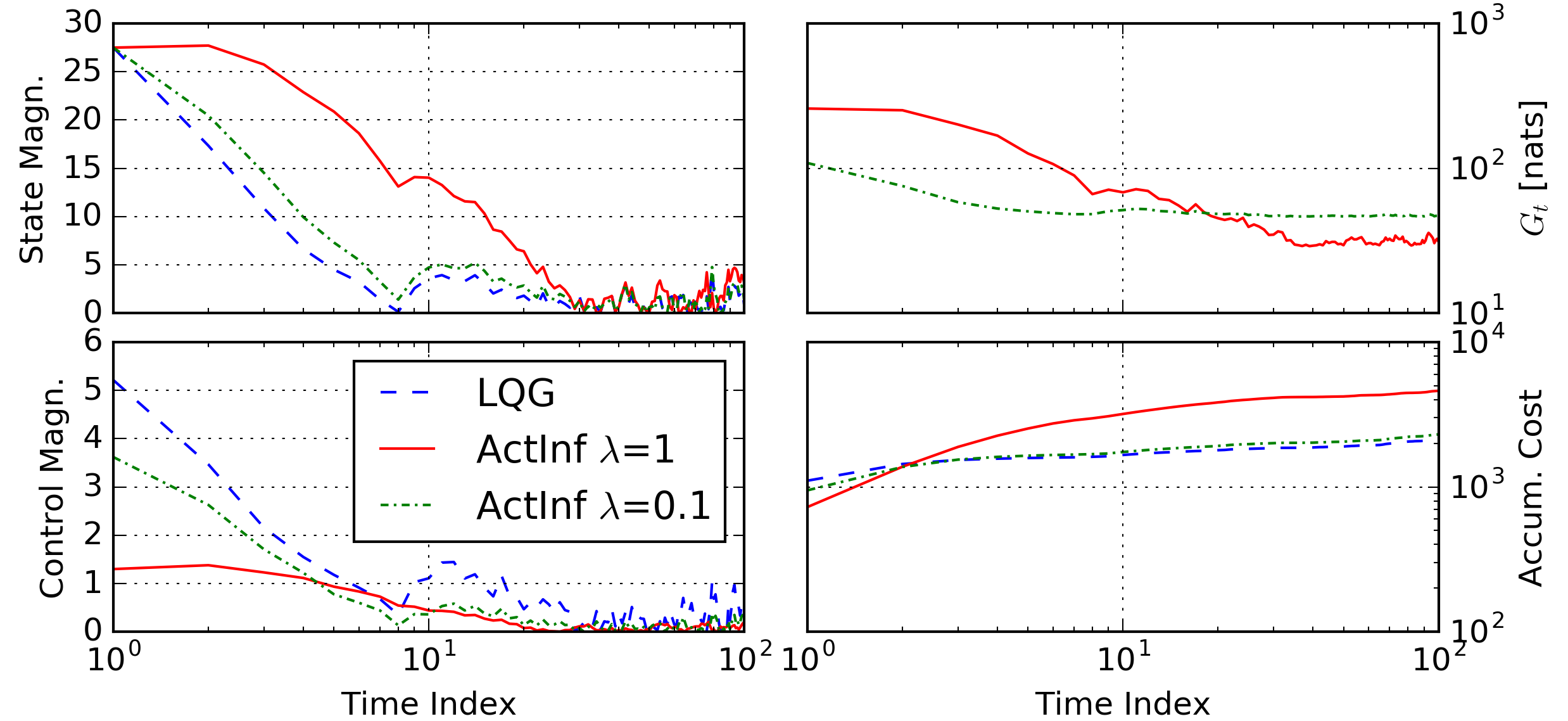}}
    \caption{Results comparing LQG with ActInf control, where the time-axis is log-scaled. The more aggressive LQG control (bottom left), leads to faster state adjustments (top left). ActInf control for small but nonzero $\lambda$ reduces to LQG control. Notably, although ActInf control with $\lambda = 1$ accumulates higher cost in terms of $\ell(x_k, u_k)$ in  \eqref{eq:quadratic_cost} (bottom right), it achieves lower free energy than ActInf control with small $\lambda$ (top right).}
    \label{fig:results}
\end{figure*}

\section{Relationship Between LQG Control And Active Inference For A Linear Gaussian SSM}
\label{sec:relationship}
We now investigate the behavior of the ActInf controller under a linear Gaussian state-space model. We assume a linear Gaussian system with the respective transition and observation precisions $W_w$ and $W_v$, as follows:
\begin{subequations}
\label{eq:lgm}
\begin{align}
    p(x_{t+1} | x_t, u_t) &= \mathcal{N}(x_{t+1} | A x_t + B u_t, W_w^{-1}) \label{eq:lgm_a}\\
    p(y_t | x_t) &= \mathcal{N}(y_t | C x_t, W_v^{-1}) \label{eq:lgm_b}\,.
\end{align}
\end{subequations}
where the notation $\mathcal{N}(z | m, V)$ represents a Gaussian distribution with the mean $m$ and the covariance (inverse precision) matrix $V = W^{-1}$ for the variable $z$. We consider the quadratic cost in \eqref{eq:quadratic_cost}, leading to independent Gaussian goal priors \eqref{eq:goal_prior}. 

\subsection{Algebraic Results for the Active Inference Controller}
A closed-form expression for the resulting ActInf regulator is obtained by propagating the messages of Fig.~\ref{fig:inference_ffg} algebraically as follows: 
\begin{thm}
	\label{thm:lqg:actinf}
	The ActInf solution for the system in \eqref{eq:lgm} is given by
	\begin{subequations}
	\label{eq:LQG_control}
	\begin{align}
	    \bm{u}_{t,q}^* &= -K_t \hat{\bm{x}}_t\\
	    K_t &= \!\left[B^{\T}\!\!\left( A \hat{V}'_t A^{\T}  +  P_{t+1}^{-1}  +  W_w^{-1}\right)^{-1}\!\!\!B + \lambda R\right]^{-1}\!\!\!\!\!\times \notag\\
	    &\quad B^{\T}\!\!\left(A \hat{V}'_t A^{\T} + P_{t+1}^{-1} + W_w^{-1}\right)^{-1}\!\!\!A\,\hat{V}'_t\,\hat{W}_t \label{eq:control_law}\\
	    \hat{V}'_t &= \left(\hat{W}_t + \lambda Q\right)^{-1}\,, \label{eq:V_hat_t}
	\end{align}
	\end{subequations}
	where $\hat{\bm{x}}_t$ and $\hat{W}_t$ are the respective mean and precision of $p_e$ (the normalized message $\smallcircled{A}$). Here, $P_k$, i.e. the precision of the backward message over state $x_k$, is given by
	\begin{subequations}
	\label{eq:LQG_recursion}
	\begin{align}
	    P_{t+T} &= \lambda Q \label{eq:P_t_T}\\
	    P_{k-1} &= -A^{\T}P_k B\left(R' + B^{\T}P_k B\right)^{-1}B^{\T}P_k A\,+ \notag\\
	    &\quad A^{\T} P_k A + \lambda Q\\
	    R' &= \left(\left[\lambda R\right]^{-1} + \left[B^{\T}W_w B\right]^{-1}\right)^{-1}. \label{eq:augmented_weighting_matrix}
	\end{align}
	\end{subequations} 
\end{thm}
\begin{proof}
	See Appendix~\ref{sec:pf:thm:lqg:actinf}. 
\end{proof}

Note that this result provides an iterative procedure for finding the ActInf solution. In particular, we initialize $P_{t+T}$ with \eqref{eq:P_t_T}. Then, $P_k$'s can be calculated iteratively and offline, i.e., without obtaining the measurements. Then, the action is found  using \eqref{eq:LQG_control}. Here, calculation of $K_t$ requires calculation of $\hat{W}_t$. In the LQG case, $p_e$ is a Gaussian pdf with a mean and covariance that are given by the standard Kalman filtering equations, see for instance \cite{korl2005factor, Kay_1993}.

We now investigate the conditions implied by the two theorems and the corollary from Sec.~\ref{sec:goal_priors}. Theorem~\ref{thm:deterministic} assumes a deterministic model and a point estimate for the current state (Cor.~\ref{cor:1}), which corresponds to $\hat{W}_t=W_w=\epsilon I_2, \epsilon\rightarrow \infty$. Theorem \ref{thm:accurateq} investigates the dependence on $\lambda$, and lets $\lambda \to 0^{+}$. In both cases \eqref{eq:augmented_weighting_matrix} reduces to $R' = \lambda R$, and \eqref{eq:control_law} reduces to $K_t = \left[B^{\T}P_{t+1}B + \lambda R\right]^{-1}\!\! B^{\T}P_{t+1}A$, thus recovering the classically optimal LQG solution in the form of the discrete-time finite horizon Ricatti equations \cite{b_GladLjung}. Note that compared to the standard LQG solution,  both $Q$ and $R$ appear to be scaled with $\lambda$ in the above equations, which has no effect on the optimal solution. This can be seen for instance by recognizing that this scaling corresponds to the scaling of both matrices with $\lambda$ in \eqref{eq:quadratic_cost}, which corresponds to a simple scalar scaling of the objective function.

\section{Numerical Results}
\label{sec:simulation}
\subsection{Scenario}
In this section, we  illustrate the performance of the ActInf controller for varying positive values of $\lambda$ and compare the results with the standard LQG scenario. The ActInf simulations are performed with the ForneyLab probabilistic programming toolbox \cite{cox_factor_2019}, and follow the experimental protocol in \cite{van_de_laar_simulating_2019}. The protocol at each time $t$ consists of three main steps: (i) find $\smallcircled{A}$ and the current state estimate $p_e$ by minimizing $\mathcal{V}_t$ \eqref{eq:FVG}, (ii) from the estimate, find a control posterior $q_t^*(u_t)$ by minimizing $\tilde{\mathcal{G}}_t$ \eqref{eq:G_tilde}, and (iii) pass a selected action to the system \eqref{eq:process} to obtain a new observation.

For the system, we use \eqref{eq:lgm}, with $C=R=Q=W_v=W_w=I_2$, $A=\begin{pmatrix}1 & 0.1\\0 & 1\end{pmatrix}$, $B=\begin{pmatrix}0.1 & 0.5\\0.05 & 0.5\end{pmatrix}$. The GM follows the system assumptions and uses a lookahead of $T=10$. We initialize the system relatively far from equilibrium, at $x_0 = (25, 25)^{\T}$ and choose a vague prior for the initial state $x_0$. 

\subsection{Discussion}

The results are presented in Fig.~\ref{fig:results}, which leads to several interesting observations. Firstly, for small but nonzero $\lambda$, the results (controls, state trajectory, the accumulated cost for $\ell(x_k, u_k)$ and also $\mathcal{G}_t$) of the ActInf controller approaches the results of the LQG controller as expected; see Sec.~\ref{sec:goal_priors} and also the discussions at the end of Sec.~\ref{sec:relationship}. We note that, for the current system, $\lambda = 0.01$ (not plotted) already renders the results of the ActInf and LQG controller nearly visually indistinguishable.

Secondly, the LQG controller is more aggressive than the ActInf regulator in terms of the controls, i.e., the magnitude of the LQG controls are relatively large compared to those of the ActInf regulator. The explicit inclusion of the process noise in ActInf is in contrast to the LQG scenario where the process and estimation noise only affect the state estimation directly but not the regulator \cite{hoffmann_linear_2017}. In particular, (\ref{eq:P_t_T}--\ref{eq:V_hat_t}) depend explicitly upon $W_w$ and $\hat{W}_t$, whereas these terms are not present in the original Ricatti equations. These terms make the ActInf controller more conservative.

Thirdly, the accumulated cost in terms of $\ell(x_k, u_k)$ for the ActInf controller approaches the optimal cost of the LQG controller under decreasing $\lambda$. This observation is consistent with Theorem~\ref{thm:accurateq}, which formulates sufficient conditions for making the ActInf solution coincide with the LQG solution. Interestingly, and perhaps counter intuitively, the terminal free energy for $\lambda=1$ is improved (lower) compared to the $\lambda=0.1$ case. This effect can be interpreted in light of the \emph{good regulator theorem}, which states that ``every good regulator of a system must be a model of that system'' \cite{conant_every_1970}. Namely, where the LQG cost function \eqref{eq:quadratic_cost} measures a quadratic cost, the free energy \eqref{eq:G_t_fact} offers an approximate measure of model fitness \eqref{eq:surprise}. This then implies that the ActInf regulator with $\lambda=1$ better models the system properties than the ActInf regulator with $\lambda=0.1$, leading to lower free energy.

\section{Conclusions}
\label{sec:conclusions}
ActInf and the free energy principle provide a flexible and general framework for stochastic optimal control problems. By including the control cost as goal priors, the control cost appears as an additive term in the free energy. The resulting free energy minimization problem can be solved by belief propagation over the associated factor graph, leading to an elegant and tractable approach to solve stochastic optimal control problems. %
In general, the ActInf controller does not solve the underlying stochastic optimal control problem. To address this, we provide sufficient conditions for which ActInf reduces to traditional stochastic optimal control. In other words, under certain conditions, stochastic optimal control is a subset of ActInf control. Finally, while it is not known for which classes of problem the sufficient conditions hold, we prove and numerically demonstrate that the ActInf controller is a generalization of the important case of the LQG controller. 

At the heart of these methods lies the fact that ActInf allows us to directly control the modeling assumptions. Therefore, we can explicitly include the anticipated effect of the costs and the noise in the control policy. Controlling these assumptions allows us to reproduce traditional stochastic optimal control solutions, such as the LQG controller.

\bibliographystyle{ieeetr}
{\bibliography{references}}


\section*{Appendix}

\appendix

\section{Proof of Theorem~\ref{thm:deterministic}}\label{sec:pf:thm:deterministic}
First, we substitute (ii), (iii) in \eqref{eq:G_t_fact} which removes the second and third term and the factors within the expectation of the last term resulting in
\begin{align}
	\mathcal{G}^{\dagger}_t[q] &= \mathbb{E}_q\!\left[\log \overline{\pi}(\overline{u}) \right] +  \lambda \mathbb{E}_{p_e, \overline{p}_t, \overline{\pi}}\!\left[\sum_{k=t}^{t+T}\ell(x_k, u_k)\right]\,. \label{eq:twoterms}
\end{align}
Let $L_t^T(\xf,x_t, \uf) \triangleq \sum_{k=t}^{t+T}\ell(x_k, u_k) $. We now focus on the second term in \eqref{eq:twoterms}
\begin{subequations}
\begin{align}
    &\lambda \int p_e(x_t)\pf(\yf, \xf | x_t, \uf) \pif(\uf) L_t^T(\xf,x_t, \uf) \d{\uf} \d{\xf} \d{\yf} \d{x_t} \\
    &= C_e \lambda \int \pp_t(\yp, \xp| \up) \pf_t(\yf, \xf | x_t, \uf) \pif(\uf) L_t^T(\xf,x_t, \uf)  \d{y} \d{x} \d{\uf} \label{eq:substitute:pe} \\
    &= C_e \lambda \int p_t(y, x| u) \pif(\uf) L_t^T(\xf,x_t, \uf) \d{y} \d{x} \d{\uf} \label{eq:substitute:p} \\
    & =C_e \lambda \mathcal{J}_t\,, \label{eq:substitute:Jt}
\end{align}
\end{subequations}
where in \eqref{eq:substitute:pe}, we have used \eqref{eq:def:ce} and \eqref{eq:def:pe}; and in \eqref{eq:substitute:p} we have used \eqref{eq:so_past_future}; and in \eqref{eq:substitute:Jt} we have used \eqref{eq:cost}, i.e. the definition of $\mathcal{J}_t$.

Since the mode of the policy is used for the current action by (iv), the negative policy entropy term $\mathbb{E}_q\!\left[\log \overline{\pi}(\overline{u}) \right] $ in \eqref{eq:twoterms} does not affect action selection. We therefore absorb the policy entropy in a constant $C$. Hence, \eqref{eq:G_t_fact} reduces to a function of the form $\mathcal{G}^{\dagger}_t[q] = \lambda C_e \mathcal{J}_t[q] + C$. Scaling of the scalar optimal control objective $\mathcal{J}_t$ does not affect regulator behavior. Hence,  the standard stochastic control solution $\bm{u}_{t,\pi}^*$ is the same as the ActInf solution $\bm{u}_{t,q}^*$. 

\section{Proof of Theorem~\ref{thm:accurateq}}
\label{sec:pf:thm:accurateq}
Recall from Theorem~\ref{thm:deterministic} that under condition (iii), the first term in \eqref{eq:G_t_fact} does not affect the optimal solution. Furthermore, (ii)  removes  the ratio $r_\lambda(\overline{y},\overline{x},x_t,\overline{u})$ from the last term in \eqref{eq:G_t_fact}. Hence, substituting in these modifications and multiplying the objective with $1/\lambda$ (note that multiplications with $1/\lambda>0$ do not change optimal solutions), we have following objective function: 
\begin{subequations}
\begin{align}
   &\frac{1}{\lambda} D(q^*(x_t) \Vert p_e(x_t)) + \frac{1}{\lambda} \mathbb{E}_{q^*}\!\!\left[\log  \frac{\phi^*(\overline{y}, \overline{x}| x_t, \overline{u})}{\overline{p}_t(\overline{y}, \overline{x} | x_t, \overline{u})}\right] + \mathbb{E}_{p_e^*, \pf^*, \pif^*}\!\left[\sum_{k=t}^{t+T}\ell(x_k, u_k)\right]\\
   &= \frac{1}{\lambda} \mathbb{E}_{q^*}\!\!\left[\log  \frac{q^*(x_t)  \phi^*(\overline{y}, \overline{x}| x_t, \overline{u})}{p_e(x_t)\overline{p}_t(\overline{y}, \overline{x} | x_t, \overline{u})}\right] + \mathbb{E}_{p_e^*, \pf^*, \pif^*}\!\left[\sum_{k=t}^{t+T}\ell(x_k, u_k)\right]\,, \label{eq:substitute:d} 
\end{align}
\end{subequations}
where in \eqref{eq:substitute:d} we have used 
\begin{align*}
    D(q^*(x_t) \Vert p_e(x_t)) \!=\!\mathbb{E}_{q^*(x_t)}\!\!\left[\log \frac{q^*(x_t)}{p_e(x_t)}\right] = \mathbb{E}_{q^*}\!\!\left[\log \frac{q^*(x_t)}{p_e(x_t)}\right]\,.
\end{align*}
Taking the limit with $\lambda \to 0^+$ and substituting (ii), we obtain   $  \mathbb{E}_{p_e^*, \pf^*, \pif^*}\!\left[\sum_{k=t}^{t+T}\ell(x_k, u_k)\right]$ as desired. Then,  the optimal control objective \eqref{eq:cost} is again (proportionally) recovered, and hence $\bm{u}_{t,q}^*=\bm{u}_{t,\pi}^*$.

\section{Proof of Theorem~\ref{thm:lqg:actinf}}
\label{sec:pf:thm:lqg:actinf}
The algebraic result for the ActInf regulator, \eqref{eq:LQG_recursion} and \eqref{eq:LQG_control}, is obtained by message passing (Fig.~\ref{fig:inference_ffg}). We derive this result by following the standard belief propagation update rules as summarized by \cite[Table 4.1]{korl2005factor}. For notational convenience, we write mean-variance and mean-precision parameterized Gaussian distributions as $\mathcal{N}_V$ and $\mathcal{N}_W$ respectively, where distribution variable arguments are left implicit.

\subsection*{Backward Recursion}
The backward recursion \eqref{eq:LQG_recursion} follows from message passing in a section of the model as visualized in Fig.~\ref{fig:recursion_ffg}. Note that our specific choice of goal prior \eqref{eq:goal_prior} is independent of $y_k$. As a result, messages $\smallcircled{2}$ and $\smallcircled{3}$ are uninformative, and do not contribute to the end result.

\begin{figure}[ht]
    \centering
    \resizebox{0.8\textwidth}{!}{

\begin{tikzpicture}
    [node distance=15mm]


    \node (x_k_min) {$\dots$};
    \node[detbox, right of=x_t_min] (eq_trans) {$=$};
    \node[box, above of=eq_trans] (N_x) {$\mathcal{N}_W$};
    \node[clamped, left of=N_x] (m_N_x) {};
    \node[clamped, right of=N_x] (W_N_x) {};
    \node[box, right of=eq_trans] (A) {$A$};
    \node[detbox, right of=A] (add) {$+$};
    \node[box, above of=add] (B) {$B$};
    \node[box, above of=B] (N_u) {$\mathcal{N}_W$};
    \node[clamped, left of=N_u] (m_N_u) {};
    \node[clamped, right of=N_u] (W_N_u) {};
    \node[box, right of=add] (N_trans) {$\mathcal{N}_W$};
    \node[clamped, above of=N_trans] (W_N_trans) {};
    \node[detbox, right of=N_trans] (eq_obs) {$=$};
    \node[box, below of=eq_obs] (C) {$C$};
    \node[box, below of=C] (N_obs) {$\mathcal{N}_W$};
    \node[clamped, left of=N_obs] (W_N_obs) {};
    \node[below of=N_obs, node distance=12mm] (y_k) {};
    \node[right of=eq_obs] (x_k) {$\dots$};

    \draw (x_t_min) -- node[anchor=south]{$x_{k-1}$} (eq_trans);
    \draw (N_x) -- (eq_trans);
    \draw (m_N_x) -- node[anchor=north]{$\rightarrow$} node[anchor=south]{$0$} (N_x);
    \draw (W_N_x) -- node[anchor=north]{$\leftarrow$} node[anchor=south]{$\lambda Q$} (N_x);
    \draw (eq_trans) -- (A);
    \draw (A) -- (add);
    \draw (B) -- (add);
    \draw (N_u) -- node[anchor=west]{$u_{k-1}$} (B);
    \draw (m_N_u) -- node[anchor=north]{$\rightarrow$} node[anchor=south]{$0$} (N_u);
    \draw (W_N_u) -- node[anchor=north]{$\leftarrow$} node[anchor=south]{$\lambda R$} (N_u);
    \draw (add) -- (N_trans);
    \draw (W_N_trans) -- node[anchor=east]{$\downarrow$} node[anchor=west]{$W_w$} (N_trans);
    \draw (N_trans) -- (eq_obs);
    \draw (eq_obs) -- (C);
    \draw (C) -- (N_obs);
    \draw (W_N_obs) -- node[anchor=north]{$\rightarrow$} node[anchor=south]{$W_v$} (N_obs);
    \draw (N_obs) -- node[anchor=west]{$\uparrow$} node[anchor=east]{$y_k$} (y_k);
    \draw (eq_obs) -- node[anchor=north]{$x_k$} (x_k);

    \msg{up}{left}{eq_obs}{x_k}{0.5}{1};
    \msg{right}{up}{N_obs}{C}{0.5}{2};
    \msg{right}{up}{C}{eq_obs}{0.5}{3};
    \msg{down}{left}{eq_obs}{N_trans}{0.5}{4};
    \msg{down}{left}{add}{N_trans}{0.5}{5};
    \msg{left}{down}{N_u}{B}{0.5}{6};
    \msg{left}{down}{B}{add}{0.5}{7};
    \msg{down}{left}{A}{add}{0.5}{8};
    \msg{down}{left}{eq_trans}{A}{0.5}{9};
    \msg{left}{down}{N_x}{eq_trans}{0.5}{10};
    \msg{down}{left}{x_k_min}{eq_trans}{0.5}{11};
\end{tikzpicture}}
    \caption{Message passing schedule for recursive backward propagation in a single (future) section of a linear Gaussian state-space model \eqref{eq:lgm}.}
    \label{fig:recursion_ffg}
\end{figure}
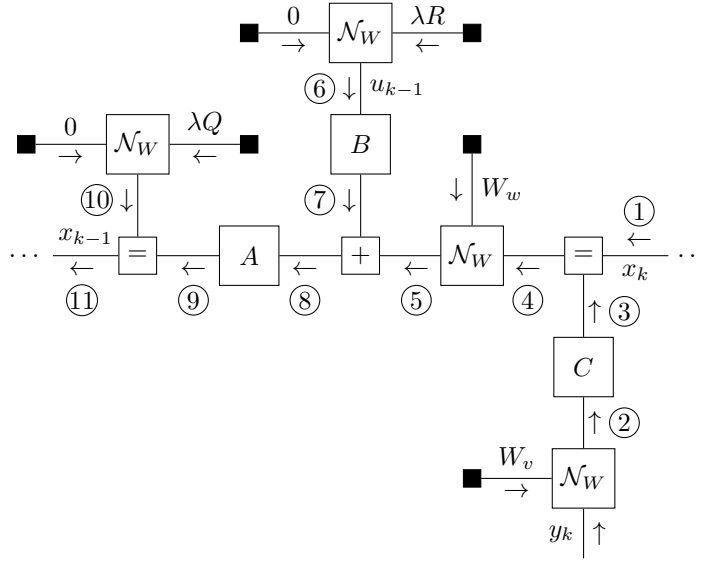

The messages of Fig.~\ref{fig:recursion_ffg} are computed as follows:
\begin{align*}
    \smallcircled{1} &\propto \NW{0, P_k}\\
    \smallcircled{2} &\propto 1\\
    \smallcircled{3} &\propto 1\\
    \smallcircled{4} &\propto \NW{0, P_k}\\
    \smallcircled{5} &\propto \NV{0, P_k^{-1} + W_w^{-1}}\\
    \smallcircled{6} &\propto \NW{0, \lambda R}\\
    \smallcircled{7} &\propto \NV{0, B(\lambda R)^{-1}B^{\T}}\\
    \smallcircled{8} &\propto \NV{0, P_k^{-1} + B(\lambda R)^{-1}B^{\T} + W_w^{-1}}\\
    \smallcircled{9} &\propto \NW{0, A^{\T}(P_k^{-1} + B(\lambda R)^{-1}B^{\T} + W_w^{-1})A}\\
    \smallcircled{10} &\propto \NW{0, \lambda Q}\\
    \smallcircled{11} &\propto \mathcal{N}_W\!\big(0, \underbrace{A^{\T}(P_k^{-1} + B(\lambda R)^{-1}B^{\T} + W_w^{-1})A + \lambda Q}_{P_{k-1}}\big)\,,
\end{align*}
where we identify a recursion over $P_{k-1}$. The recursion is initialized with $P_{t+T} = \lambda Q$, and terminates when $P_{t+1}$ is computed. The result simplifies further:
\begin{align*}
    P_{k-1} &= A^{\T}(P_k^{-1} + B\overbrace{[(\lambda R)^{-1} + (B^{\T}W_w B)^{-1}]}^{[R']^{-1}}B^{\T})A\\
    &\quad + \lambda Q \qquad \text{\emph{(using }} BB^{-1} = I \text{\emph{)}}\\
    &= A^{\T}P_k A - A^{\T}P_k B[R' + B^{\T}P_k B]^{-1}B^{\T}P_k A\\
    &\quad + \lambda Q \qquad \text{\emph{(using Woodbury identity)}}\,,
\end{align*}
with
\begin{align*}
    R' &= [(\lambda R)^{-1} + (B^{\T}W_w B)^{-1}]^{-1}\,.
\end{align*}
This concludes the derivation of \eqref{eq:LQG_recursion}. The computation of $R'$ can be simplified by making use of Searle's identity.

\subsection*{Control Law}
The control law \eqref{eq:LQG_control} follows from message passing in a section of the model as visualized in Fig.~\ref{fig:control_ffg}.

\begin{figure}[ht]
    \centering
    \resizebox{0.8\textwidth}{!}{

\begin{tikzpicture}
    [node distance=15mm]


    \node (x_t) {$\dots$};
    \node[detbox, right of=x_t_min] (eq_trans) {$=$};
    \node[box, above of=eq_trans] (N_x) {$\mathcal{N}_W$};
    \node[clamped, left of=N_x] (m_N_x) {};
    \node[clamped, right of=N_x] (W_N_x) {};
    \node[box, right of=eq_trans] (A) {$A$};
    \node[detbox, right of=A] (add) {$+$};
    \node[box, above of=add] (B) {$B$};
    \node[box, above of=B] (N_u) {$\mathcal{N}_W$};
    \node[clamped, left of=N_u] (m_N_u) {};
    \node[clamped, right of=N_u] (W_N_u) {};
    \node[box, right of=add] (N_trans) {$\mathcal{N}_W$};
    \node[clamped, above of=N_trans] (W_N_trans) {};
    \node[detbox, right of=N_trans] (eq_obs) {$=$};
    \node[box, below of=eq_obs] (C) {$C$};
    \node[box, below of=C] (N_obs) {$\mathcal{N}_W$};
    \node[clamped, left of=N_obs] (W_N_obs) {};
    \node[below of=N_obs, node distance=12mm] (y_t_plus) {};
    \node[right of=eq_obs] (x_t_plus) {$\dots$};

    \draw (x_t_min) -- node[anchor=south]{$x_t$} (eq_trans);
    \draw (N_x) -- (eq_trans);
    \draw (m_N_x) -- node[anchor=north]{$\rightarrow$} node[anchor=south]{$0$} (N_x);
    \draw (W_N_x) -- node[anchor=north]{$\leftarrow$} node[anchor=south]{$\lambda Q$} (N_x);
    \draw (eq_trans) -- (A);
    \draw (A) -- (add);
    \draw (B) -- (add);
    \draw (N_u) -- node[anchor=west, xshift=8mm]{$u_t$} (B);
    \draw (m_N_u) -- node[anchor=north]{$\rightarrow$} node[anchor=south]{$0$} (N_u);
    \draw (W_N_u) -- node[anchor=north]{$\leftarrow$} node[anchor=south]{$\lambda R$} (N_u);
    \draw (add) -- (N_trans);
    \draw (W_N_trans) -- node[anchor=east]{$\downarrow$} node[anchor=west]{$W_w$} (N_trans);
    \draw (N_trans) -- (eq_obs);
    \draw (eq_obs) -- (C);
    \draw (C) -- (N_obs);
    \draw (W_N_obs) -- node[anchor=north]{$\rightarrow$} node[anchor=south]{$W_v$} (N_obs);
    \draw (N_obs) -- node[anchor=west]{$\uparrow$} node[anchor=east]{$y_{t+1}$} (y_t_plus);
    \draw (eq_obs) -- node[anchor=north]{$x_{t+1}$} (x_t_plus);

    \msg{up}{left}{eq_obs}{x_t_plus}{0.5}{1};
    \msg{right}{up}{N_obs}{C}{0.5}{2};
    \msg{right}{up}{C}{eq_obs}{0.5}{3};
    \msg{down}{left}{eq_obs}{N_trans}{0.5}{4};
    \msg{down}{left}{add}{N_trans}{0.5}{5};
    \msg{left}{down}{N_u}{B}{0.5}{6};
    \msg{down}{right}{x_t}{eq_trans}{0.5}{7};
    \msg{left}{down}{N_x}{eq_trans}{0.5}{8};
    \msg{down}{right}{eq_trans}{A}{0.5}{9};
    \msg{down}{right}{A}{add}{0.5}{10};
    \msg{right}{up}{B}{add}{0.5}{11};
    \msg{right}{up}{B}{N_u}{0.5}{12};
\end{tikzpicture}}
    \caption{Message passing schedule for the control law in a single (present) section of a linear Gaussian state-space model \eqref{eq:lgm}.}
    \label{fig:control_ffg}
\end{figure}
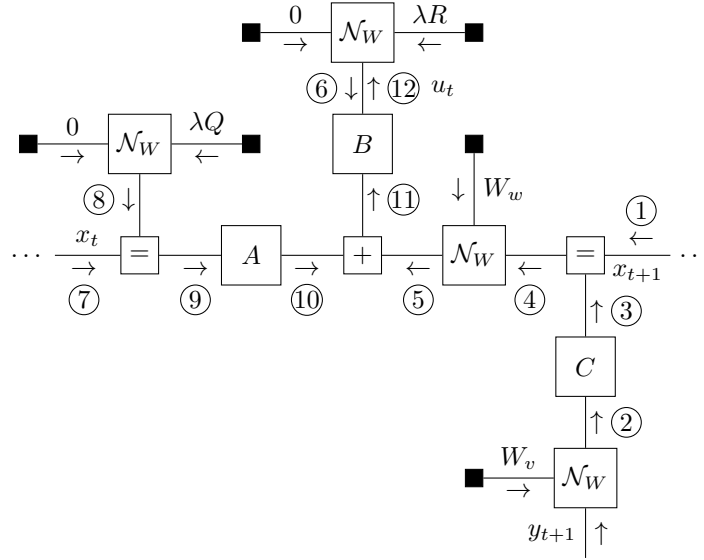

The messages of Fig.~\ref{fig:control_ffg} are computed as follows:
\begin{align*}
    \smallcircled{1} &\propto \NW{0, P_{t+1}}\\
    \smallcircled{2} &\propto 1\\
    \smallcircled{3} &\propto 1\\
    \smallcircled{4} &\propto \NW{0, P_{t+1}}\\
    \smallcircled{5} &\propto \NV{0, P_{t+1}^{-1} + W_w^{-1}}\\
    \smallcircled{6} &\propto \NW{0, \lambda R}\\
    \smallcircled{7} &\propto \NW{\hat{\bm{x}}_t, \hat{W}_t}\\
    \smallcircled{8} &\propto \NW{0, \lambda Q}\\
    \smallcircled{9} &\propto \NV{[\hat{W}_t + \lambda Q]^{-1}\hat{W}_t \hat{\bm{x}}_t, [\hat{W}_t + \lambda Q]^{-1}}\\
    \smallcircled{10} &\propto \NV{A[\hat{W}_t + \lambda Q]^{-1}\hat{W}_t \hat{\bm{x}}_t, A[\hat{W}_t + \lambda Q]^{-1}A^{\T}}\\
    \smallcircled{11} &\propto \mathcal{N}_V\!\Big(-A[\hat{W}_t + \lambda Q]^{-1}\hat{W}_t \hat{\bm{x}}_t,\\
    &\quad A[\hat{W}_t + \lambda Q]^{-1}A^{\T} + P_{t+1}^{-1} + W_w^{-1}\Big)\\
    \smallcircled{12} &\propto \mathcal{N}_W\!\Big(-B^{-1}A[\hat{W}_t + \lambda Q]^{-1}\hat{W}_t \hat{\bm{x}}_t,\\
    &\quad B^{\T}[A[\hat{W}_t + \lambda Q]^{-1}A^{\T} + P_{t+1}^{-1} + W_w^{-1}]^{-1}B \Big)\,.
\end{align*}

The current control then follows from
\begin{align*}
    q^*_t(u_t) &\propto \smallcircled{6}\times \smallcircled{12}\\
    \bm{u}_t &= \operatorname{mode} q^*_t(u_t)\\
    &= -K_t \hat{\bm{x}}_t\,,
\end{align*}
where (using the Gaussian equality rule)
\begin{align*}
    K_t &= [B^{\T}( A \hat{V}'_t A^{\T}  +  P_{t+1}^{-1}  +  W_w^{-1})^{-1}B + \lambda R]^{-1}\times\\
    &\quad B^{\T}(A \hat{V}'_t A^{\T} + P_{t+1}^{-1} + W_w^{-1})^{-1}A\,\hat{V}'_t\,\hat{W}_t\,,
\end{align*}
with
\begin{align*}
    \hat{V}'_t &= (\hat{W}_t + \lambda Q)^{-1}\,.
\end{align*}
This concludes the derivation of \eqref{eq:LQG_control}.

\end{document}